\documentclass[a4paper,11pt,USenglish,cleveref]{article}
\usepackage{fullpage}

\usepackage[T1]{fontenc}
\usepackage{textcomp}
\usepackage[english]{babel}
\usepackage{amsmath,amsthm,amssymb}
\usepackage{color,enumerate,array,graphicx,multirow,tabularx,listings}
\usepackage[left,mathlines]{lineno}
\usepackage{algorithm,float,algorithmicx}
\usepackage[noend]{algpseudocode}
\usepackage{tikz}
\usepackage{pgf}
\usepackage{tkz-euclide}
\usetikzlibrary{positioning}
\usepackage{pgfplots}
\usepackage{ifthen}
\usepackage{forest}

\usetikzlibrary{graphs}
\usetikzlibrary {graphs.standard} 

\usepackage{todonotes}
\usepackage{hyperref}
\usepackage{cleveref}
\newcommand{\citetodo}[1]{\textcolor{red}{\cite{??}}\todo[color=orange]{\tiny Missing reference}}

\usepackage{float}
\newtheorem{theorem}{Theorem}[section]
\newtheorem{lemma}{Lemma}[section]
\newtheorem{corollary}{Corollary}[section]
\newtheorem{proposition}{Proposition}[section]

\newtheorem{claim}{Claim}[section]
\newcommand{\toto}{xxx}

\newcommand{\m}[1]{\ensuremath{\mathcal{#1}}}
\algnewcommand\algorithmiccase{\textbf{case}}
\algnewcommand\StateCase[1]{\State\hphantom{x}\ #1 \algorithmicthen} 
\algdef{SE}[CASE]{Case}{EndCase}[1]{\algorithmiccase\ #1\ \algorithmicthen}{\algorithmicend\ \algorithmiccase}%
\algdef{SE}{Begin}{End}{\textbf{begin}}{\textbf{end}}
\algnewcommand{\IfL}[1]{\State\algorithmicif\ #1\ \algorithmicthen}
\algnewcommand{\EndIfL}{\unskip\ \algorithmicend\ \algorithmicif}

\newcommand*\patchAmsMathEnvironmentForLineno[1]{%
  \expandafter\let\csname old#1\expandafter\endcsname\csname #1\endcsname
  \expandafter\let\csname oldend#1\expandafter\endcsname\csname end#1\endcsname
  \renewenvironment{#1}%
     {\linenomath\csname old#1\endcsname}%
     {\csname oldend#1\endcsname\endlinenomath}}%
\newcommand*\patchBothAmsMathEnvironmentsForLineno[1]{%
  \patchAmsMathEnvironmentForLineno{#1}%
  \patchAmsMathEnvironmentForLineno{#1*}}%
\AtBeginDocument{%
  \patchBothAmsMathEnvironmentsForLineno{equation}%
  \patchBothAmsMathEnvironmentsForLineno{align}%
  \patchBothAmsMathEnvironmentsForLineno{flalign}%
  \patchBothAmsMathEnvironmentsForLineno{alignat}%
  \patchBothAmsMathEnvironmentsForLineno{gather}%
  \patchBothAmsMathEnvironmentsForLineno{multline}}

\title{The Computational Power of Distributed Shared-Memory Models with Bounded-Size Registers\thanks{All authors are supported by the ANR Project DUCAT (ANR-20-CE48-0006)}}

\author{
Carole Delporte\\
{\small IRIF}\\
{\small Université Paris Cité and CNRS}\\
{\small France}
\and
Hugues Fauconnier\\
{\small IRIF}\\
{\small Université Paris Cité and CNRS}\\
{\small France}\and
Pierre Fraigniaud\\
{\small IRIF}\\
{\small Université Paris Cité and CNRS}\\
{\small France}\and
Sergio Rajsbaum\\
{\small Instituto de Matematicas}\\
{\small UNAM}\\
{\small Mexico}
\and
Corentin Travers\\
{\small LIS}\\
{\small Aix-Marseille Univ. and CNRS}\\
{\small France}
}

\begin{document}
\date{}
\maketitle

\begin{abstract}
The celebrated  \emph{Asynchronous Computability Theorem} of Herlihy and Shavit (STOC~1993 \& STOC~1994) provided a topological characterization of the
tasks that are solvable in a distributed  system where processes are communicating by writing and reading shared registers, and where any number of processes can fail by crashing. 
However, this characterization  assumes the use of \emph{full-information} protocols,  that is, protocols in which each time any of the  processes writes in the shared memory, it communicates everything it learned  since the beginning of the execution. 
Thus, the characterization  implicitly assumes that each register in the shared memory is of \emph{unbounded} size. 
Whether unbounded size registers
are unavoidable for the model of computation to be \emph{universal} is the central question studied in this paper. Specifically, is any task that is solvable using unbounded registers solvable using registers of bounded size?
More generally, when at most $t$ processes can crash, is the model with bounded
size registers universal? These are the questions answered in this paper. 

We show that the two extreme models, i.e., the 1-\emph{resilient} model (at most one process can crash), and the \emph{wait-free} model (any number of processes can crash) behave quite differently. Indeed, \emph{constant} size registers are universal for 1-resilient computation, whereas constant size registers are not universal for wait-free computation (with more than two processes). In fact, wait-free computing with bounded registers is not universal even if the size of the registers  may depend on the number $n$ of processes (for $n>2$). 
The more refined picture is as follows.  If the number $t$ of failures is assumed to be less than half of the number of processes, then registers of size $O(t)$ bits are sufficient to solve any task that is solvable by a full-information protocol. 
 The situation is radically different when more than half of the processes may crash,
 which includes the wait-free case, 
 as we show that bounded registers are \emph{not} universal. Specifically, we show that, for any bound $f:\mathbb{N}\to \mathbb{N}$, and for any $n>2$, there exists a task that cannot be solved by $n$ processes communicating via registers of size $f(n)$ bits,
 but is solvable by $n$ processes using unbounded registers.
Finally, we consider the \emph{iterated} shared memory model, which has played a crucial role since Herlihy and Shavit's characterization.
We show that 1-bit registers (per iteration)  are sufficient for solving wait-free any task that is solvable by a full-information wait-free protocol. Therefore, while the wait-free read/write shared-memory model  is known to be computationally equivalent to the wait-free iterated immediate snapshot (IIS) model, this equivalence only holds if registers are unbounded. 
 \end{abstract}

\thispagestyle{empty}
\newpage
\setcounter{page}{1}


\section{Introduction}
\label{sec:intro}

\subsection{Objective}

In sequential computing, universal models such as the Turing machine were identified from the 
early days of modern computer science. 
In \emph{distributed} computing, the power of a model may however depend
on the reliability of its components, and on other assumptions. Even for the most
basic standard asynchronous read/write \emph{shared memory} model, where processes may crash,
the problem is not fully resolved. In this model, the system is composed of $n$ processes that are assumed to 
be \emph{asynchronous}, meaning that 
each one runs at its own  speed, which may evolve with time. A process is a deterministic state machine
that may fail by \emph{crashing}, i.e. stopping its execution at any moment.
Each process has a  \emph{register} in the shared memory, in which it can write information that 
can be read by all other processes. Such a register is thus called single-writer/multiple-reader (SWMR).

The SWMR registers forming the shared memory, one for each process, are supposed to be of unbounded size, i.e., there are no a priori limit on the amount of information that can be stored in a register~\cite{AttiyaLS94,HerlihyS99}.
This  is appropriate for studying the limitations of distributed computing. For example, the
\emph{consensus} task is not {wait-free} solvable~\cite{Herlihy91,LA87}, where 
an algorithm is said to be \emph{wait-free} if it tolerates crash failures by  any  set of at most $n-1$ processes. 
The terminology ``wait-free''  comes from the fact that a (correct) process must not ``wait'' for other processes, merely because any of them may be slow, but may also have crashed. It is actually known 
that consensus is impossible even if at most one process can crash~\cite{FischerLP85}.
This illustrates that distributed computing experiences inherent coordination limitations, as consensus is impossible even 
if each process is an infinite state (deterministic) machine, and even if each shared register is of unbounded size.

The celebrated \emph{asynchronous computability theorem}~\cite{HerlihyS99} provided a \emph{topological} characterization of the
tasks that are wait-free solvable in a distributed  system.
This result assumes \emph{full-information} protocols, that is, protocols in which each time   any of the $n$  processes writes in its   register, it writes everything it learned  since the beginning of the execution, and each process remembers all the values it has read from the registers of the other processes. 
 The characterization naturally assumes that each register in the shared memory  is of unbounded size
because  there are tasks that require arbitrary many read and write operations to be solved, and it is desirable that information previously written by a process in its register is not erased by later writes, as this information may not have been read yet by other processes. 
We consider the wait-free shared-memory model 
because it is at the center of the distributed computability theory. In particular, for every number $n>2$ of processes, it is undecidable whether a task is wait-free solvable or not~\cite{GafniK99,HerlihyR97}.
The  topological properties of the wait-free model can often be generalized to
more powerful models (e.g., $t$-resilient models), in more tasks can be solved.

The question addressed in this paper is the following: 
Are there distributed tasks that can be solved in the wait-free shared-memory
model assuming unbounded registers, but that cannot be solved in 
the wait-free shared-memory model whenever the size of the registers is bounded?
That is, we question the existence of a \emph{universal} model for 
wait-free shared-memory computing,  with bounded-size registers. 
We answer this question for the wait-free model, and more generally, for 
the \emph{$t$-resilient} version of the model,
where the number of processes that can crash during any execution is at 
most~$t$, for a fixed $t\in\{1,\dots,n-1\}$. 
In other words, the purpose of this article is to answer the following question:

\begin{center}
\begin{minipage}{14cm}
For which degree of crash resiliency the asynchronous read/write shared-memory model with bounded-size registers remains universal? 
\end{minipage}
\end{center}

\subsection{Our Results}

In a nutshell, we  show that bounded size registers yield a  universal model for read/write shared-memory computing only when less than half of the processes may fail. More specifically,  
 we first show that wait-free computing with bounded registers is not universal, for $n>2$ processes. 
  This actually holds even for the $t$-resilient model whenever $t>\frac{n}{2}$. Moreover, the result holds in a strong sense, that is, even if one allows the size of the registers to depend on the number $n$ of processes in the system. 

\begin{theorem}\label{main:wait-free-not-universal}
For any function $f:\mathbb{N}\to\mathbb{N}$, for any $n>2$, and for any $t$ with $\frac{n}{2} < t < n$, there exists a task  that is solvable by $n$ processes in the $t$-resilient shared-memory model with unbounded-size registers, but that cannot be solved by $n$ processes in the $t$-resilient shared-memory model with registers of at most $f(n)$ bits. 
\end{theorem}

We stress that the task whose existence is established in Theorem~\ref{main:wait-free-not-universal}  has binary inputs, i.e.,  each process starts with an input values in $\{0,1\}$. The non solvability of the task is therefore not due to large-size inputs that the processes may need to write to small registers. Theorem~\ref{main:wait-free-not-universal} holds for every $n>2$. The case $n=2$, which is the smallest distributed system one can conceive, and where 1-resilient computing and wait-free computing coincide, actually displays a radically different behavior, as we show the following.

\begin{theorem}\label{theo:case-of-two-proc}
Any task solvable by 2 processes in the wait-free shared-memory model with unbounded-size registers can be solved  by 2 processes  in the wait-free shared-memory model with 1-bit registers. 
\end{theorem}

Theorem~\ref{theo:case-of-two-proc} holds  for arbitrary tasks, even those with arbitrarily large inputs, by assuming that each of the two processes has access to a special write-once register in which it can write its input, which can  be read by the other process. 
These registers are solely used for exchanging inputs, and cannot be used for coordination between the processes, which is the core of  the  distributed computation. Using $1$-bit registers however comes at the price of an exponential slowdown compared to using unbounded registers. Nevertheless, we also show that this overhead can actually be eliminated when using constant-size registers, just that the constant needs to be slightly larger than~1.

We show that the threshold in Theorem~\ref{main:wait-free-not-universal} regarding the maximum number~$t$ of failures is essentially the best one can expect. Indeed, for $t<\frac{n}{2}$, registers of bounded size are universal, as shown by the following theorem.
 
 \begin{theorem}\label{theo:small-values-of-t}
 There exists $\beta\geq 1$ such that, for any $n\geq 2$, and any $t$ with $1 \leq t < \frac{n}{2}$, every task solvable by $n$ processes in the $t$-resilient shared-memory model with unbounded-size registers can be solved  by $n$ processes  in the $t$-resilient shared-memory model with registers of $\beta\, t$ bits.  
\end{theorem}

Note that the constant $\beta$ does not depend on the task,  neither on the number of processes. 
 In particular, every task solvable by $n$ processes in the $1$-resilient shared-memory model with unbounded-size registers can be solved  by $n$ processes  in the $1$-resilient shared-memory model with constant size registers. 
  Figure~\ref{fig:summary} summarizes all these results. 
 
 \begin{figure}[tb]
 \centering
\includegraphics[width=10cm]{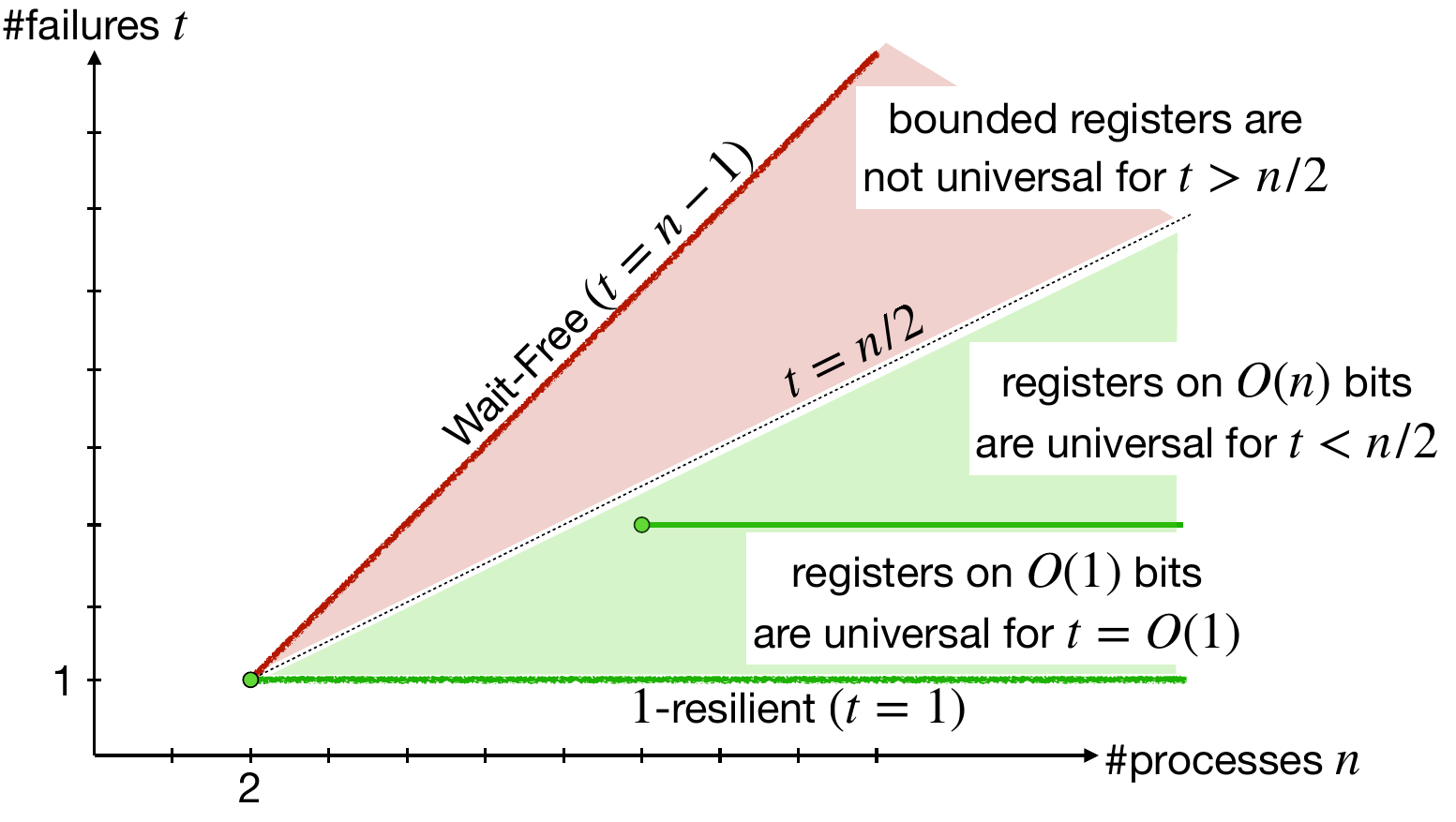}
\caption{\sl Summary of our results. For $t>\frac{n}{2}$, the $t$-resilient model with bounded registers is not universal, even if the bound on the size of the registers may depend on the number of processes (Theorem~\ref{main:wait-free-not-universal}). For $t<\frac{n}{2}$, the $t$-resilient model with registers on $O(n)$ bits is universal. Actually, it is universal even with registers on $O(t)$ bits (Theorem~\ref{theo:small-values-of-t}). Therefore, for $t=O(1)$, the $t$-resilient model with constant-size registers is universal. Moreover, for the case $n=2$, where 1-resilient computing and wait-free computing coincide, 1-bit registers are universal (Theorem~\ref{theo:case-of-two-proc}). }
\label{fig:summary}
\end{figure}

Finally, we conclude the paper by considering the so-called \emph{iterated} version of the wait-free model, 
 which is an  artificial abstraction  that has been used in most papers  since the asynchronous computability theorem~\cite{HerlihyS99} was established, mainly because it greatly simplifies topological arguments~\cite{bookHerlihyKR2013}. 
This abstraction assumes that, instead of a single array $M$ of $n$ registers, one register per process, 
the  shared-memory consists of an infinite sequence $(M_r)_{r\geq 1}$ of such arrays.  
The computation is then  organized in rounds. At every round $r\geq 1$, the processes communicate by writing once, and then
reading the registers in the $r$th memory $M_r$.
Several variants of iterated models have been studied, some even more powerful than wait-free~\cite{ImbsRV15,KuznetsovR20, KuznetsovRH18},
but the general 
idea is that the restriction of writing only once each memory $M_r$, and the inability of re-reading a ``lower'' memory $M_s$ for $s<r$, 
can be done without loss of generality, as far as task solvability is concerned.
Specifically, in the wait-free model, a task is solvable  if and only if it is solvable in the corresponding iterated
model~\cite{BorowskyG97simple,GafniR10}, both using unbounded registers. 
We show that, in fact, bounded size registers suffice in the iterated model.  
 
 \begin{theorem}\label{theo:iterated-immediate-snapshot}
 For any $n\geq 2$, any task solvable wait-free by $n$ processes in the iterated model with unbounded registers is solvable  wait-free  in the  iterated model with 1-bit registers (at each level $r$ of the memory).
\end{theorem}

Therefore, while the wait-free read/write shared-memory model is known to be computationally
equivalent to its wait-free iterated version, this equivalence holds only if
registers are unbounded. 


\subsection{Related Work}
\label{sec:related}

The fundamental paper by Fischer, Lynch, and Paterson~\cite{FischerLP85} showed 
that  Turing computability theory is not sufficient for analyzing computability in asynchronous, distributed systems, by proving that the consensus task is not solvable 
 in a message passing system even if only one process may fail by crashing,
 even if each process is an infinite state machine.
 Biran, Moran and Zaks~\cite{BMZ90charac} extended the result into a full
  characterization of the tasks that are solvable in this model of computation.
   Later, it was  shown by Attiya, Bar{-}Noy and Dolev~\cite{AttiyaBD95} 
   that the message passing and shared memory models are equivalent with respect
   to the tasks that they can solve,  when the number of processes that can crash $t$ is
   $t<n/2$.
  
 Biran, Moran and Zaks~\cite{BMZ90charac} showed that any $1$-resilient solvable task
 can be solved by a form of  approximate agreement on a graph, defined by
 the possible outputs of the task.  It is ``approximate'' in the sense that processes
 decide on the same vertex or on adjacent vertices of the graph. 
 In Section~\ref{sec:equality} we show how  this approach also applies to the case of two processes communicating with constant size registers. 
For $t<n/2$, we prove the universality of registers of size $O(t)$ in Section~\ref{sec:minority},
transforming the algorithm with unbounded registers to a message passing algorithm 
using the above simulation~\cite{AttiyaBD95}, then reducing the
 communication links between processes, as long as the network is $(t+1)$-connected,
and finally relying on the  alternating bit protocol (see~\cite{bartlett1969note,lynch1968computer}).

The approximate agreement problem where processes start with real valued inputs
in a range $[a,b]$ was introduced by Dolev, Lynch, Pinter, Stark, and Weihl~\cite{DolevLPSW86}, and since then it has been thoroughly studied, providing
solutions in various  models, both synchronous and asynchronous, 
 even when  processes can exhibit arbitrary,  Byzantine failures, e.g.~\cite{AbrahamAD04}.
Closer to our work, Hoest and Shavit~\cite{HoestS06} considered the problem in a shared-memory
iterated model of computation, proving matching upper and lower bounds of
$\log_d (b-a)/\epsilon$ on number of steps taken by a process, where $d=3$ for two processes, and $d=2$ for three or more processes.
  A similar result was proved in message passing, dynamic graph models~\cite{FuggerNS21}.
 We show in Section~\ref{sec:equality} that approximate agreement can be
 solved with registers of constant size by two processes with the same step complexity,
 through a novel technique that considers protocol complexes that are not full-information.

 In contrast,  when a majority of the processes
 may fail, restricting the size of the registers strictly limits the degree of 
 approximation that can be achieved. In Section~\ref{sec:no_dum} (sketch in Section~\ref{sec:sketch}).
We show that  there exists  
 a sufficiently small $\epsilon$,  such that approximate agreement cannot be solved by
$n$ processes in the $t$-resilient shared-memory model with registers of $O(n)$ bits.
  
 Generalizing the  task solvability characterization~\cite{BMZ90charac} from $t=1$ to $t=n-1$
 was done by Herlihy and Shavit~\cite{HerlihyS99} using combinatorial topology,
 essentially going from 1-dimensional graphs needed to model two process
 executions, to $(n+1)$-dimensional simplicial complexes needed to model $n$ process
 executions. In particular, an $n$-dimensional form of approximate agreement
 task played the analogue of the approximate agreement graphs task used in~\cite{BMZ90charac}.
 An extensive literature has investigated the solvability of particular tasks, such as consensus, 
 approximate agreement~\cite{DolevLPSW86},
 set agreement~\cite{Raynal16b}, and renaming~\cite{CastanedaRR11}.
For an overview of results in this
area, and why it is intimately related to topology, see~\cite{bookHerlihyKR2013}.
Most of the work uses iterated models, but a few have studied  non-iterated models e.g.~\cite{AttiyaR02}.
 
Analyzing the topology of the  protocol complex representing
the executions of an algorithm, proved to be much easier in an iterated
model, which has the same task computability power~\cite{BorowskyG97simple}.  
Dynamic network models where synchronous processes that do not crash communicate by
sending messages that may be lost are a closely related to iterated shared memory models~\cite{AfekG15,GodardP20,FuggerNS21}.
It was shown that, for two processes, it is possible to solve any wait-free 
solvable task using  $1$-bit messages, without incurring any cost in the
optimal number of communication rounds~\cite{Delporte-Gallet20}.
In Section~\ref{sec:iterated} we generalize this result for any $n\geq 2$, showing
that any task solvable wait-free by $n$ processes in the IIS model with unbounded registers is solvable  wait-free  in the  IIS model with 1-bit registers (per round).
For this, we  use the snapshot implementation of~\cite{BG92}.

For the case of algorithms with constant size registers,
the input registers allow processes to communicate their inputs for free,
to separate the concerns about the size of the
inputs from the question about reducing the costs of the full information protocol under
asynchrony and failures.
For example, bounds on approximate agreement often depend both on $\epsilon$ and
on the magnitude or the spread of the inputs e.g.~\cite{AttiyaE22,Herlihy91,Moran95,Schenk95}.
 Instead,
we separate the aspects related to the number of possible inputs, and the
actual approximate agreement computation,   by assuming instead  inputs  $\{0,1\}$,
and then use a solution to this binary version, to solve any task.
In a sense, communicating the input is an information theoretic concern, and indeed,
the problem has been thoroughly studied from this perspective, e.g. Orlitsky~\cite{OrlitskyV03},
and more recently~\cite{Rajsbaum23},
and this is the case also in the wait-free setting~\cite{Delporte-Gallet20}.
Indeed, Schenk~\cite{Schenk95} (see also \cite[Theorem~3.11]{Schenk96} for more details) shows that any two process
protocol for approximate agreement communicating through $b$ bit MWMR registers
must use a number of registers that grows as $\epsilon$ gets smaller,
an apparent contradiction with our solution using two $1$-bit registers. The reason is
precisely that Schenk assumes  inputs in $[0,1]$, 
while we assume only two possible inputs. Attiya and Ellen~\cite{AttiyaE22} focus on
multi-dimensional version of approximate agreement (and MWMR registers), but
in the one-dimensional case,  achieve an approximate agreement algorithm whose step complexity is $O(\min\{ \log n(\log n+\log(S/ \epsilon)),\log(M/\epsilon)\})$, where $S$ is the spread $b-a$
and $M$ is their magnitude, the largest of the absolute values of the inputs.
 In contrast, our lower bounds
for more than 2 processes showing the impossibility of constant (or bounded) registers
size hold even if processes have only two possible inputs. Schenk considers also the wait-free
case of  $n>2$, and 
 proves upper and lower bounds of $\Theta(\log (s/\epsilon))$ single-bit MWMR registers,
 for real-valued inputs of magnitud at most $s$, while we show a solution with registers
 whose size does not depend on $\epsilon$, for $t<n/2$, and furthermore, in the 
 case of a majority of failures 
 we show that there must be a dependence on $\epsilon$, even in the case of only two possible inputs.

\subsection{Paper Organization}

In Section~\ref{sec:model}, we formally describe our model, and we recall basic facts from the distributed computing framework that will be used for establishing our own results. Then, in Section~\ref{sec:sketch}, we provide the reader with some intuition about the main arguments used for proving our results. The next four sections are dedicated to the proofs of the four Theorems~\ref{main:wait-free-not-universal}-\ref{theo:iterated-immediate-snapshot}, respectively. Section~\ref{sec:faster-eps-agr} revisits Theorem~\ref{theo:case-of-two-proc}, and demonstrates that, by slightly increasing the size of the registers, from 1~bit to $O(1)$ bits, the simulation used in the proof of Theorem~\ref{theo:case-of-two-proc} can be speeded up exponentially, and running at  essentially the same time as the algorithms using unbounded registers. 
Finally, Section~\ref{sec:concl} concludes the paper.


\section{Model and Basic Facts}
\label{sec:model}

We use a standard model of computation that is briefly summarized below, together with reminders of classical results in distributed computing that will be useful in the paper. More details can be found in~\cite{attiya2004distributed,raynalBook13}. 

\paragraph{Computing Model.}

We consider  $n\geq 2$ deterministic  processes, modeled as (possibly infinite) state machines labeled from~1 to~$n$, communicating by reading and writing to a shared memory.

The shared memory consists of $n$ \emph{registers} $R_1,\dots,R_n$. 
For every $i\in\{1,\dots,n\}$, register $R_i$ is the only register in which process~$i$, abbreviated~$p_i$, can write. However, $p_i$ can read all  registers. Such registers are called single-writer multi-reader (SWMR) registers. 
Notice that the assumption of having only one SWMR register per process is done without loss of generality, in the sense that a  SWMR register can emulate a constant number of SWMR registers.
When a process $i$ writes a value~$v$ in its register $R_i$, the content of the register is erased and replaced by~$v$.  Each access to a register, whether it be a read or a write, is atomic, and it is assumed that two concurrent accesses to a same register never occur. Registers are often referred to as \emph{objects} in the literature on shared-memory distributed computing, and the standard syntax for a write operation performed by a process in a register $R$ is $R.\mathsf{write}(v)$, meaning that the value~$v$ is written in register~$R$. Similarly, $R.\mathsf{read}()$ returns the value~$v$ currently stored in register~$R$. 

An atomic operation performed by a process in the shared memory is called a \emph{step}. The processes are \emph{asynchronous}, and may fail by \emph{crashing}. When a process crashes, it stops taking any step, stops performing any internal computation, and never recovers. 

\paragraph{Tasks and Algorithms.}

A  \emph{task} is the basic form of distributed computing problem, the analogous of a function for Turing machines, studied since~\cite{BMZ90charac}. It is defined by a set of input  \emph{configurations} of the $n$-process system, a set of output configurations,  and a mapping specifying, for every input configuration, the set of output configurations that are legal with respect to this input. An algorithm solves a given task if each of the $n$ processes can start with any private input value taken from the input domain of the task, communicate with the other processes, and eventually halt with a private output value such that the collection of output values produced by the processes is consistent with the collection of input values with respect to the input-output specification of the task. 

More formally, a task for $n$ processes is defined as a triplet $\Pi =  (\m{I},\m{O},\Delta)$ where $\m{I}$ and $\m{O}$ are finite sets of  input and output configurations (usually defined as simplicial complexes), and   $\Delta:\m{I}\to 2^\m{O}$ maps each input configuration $\sigma \in \m{I}$ to a set of  output configurations $\Delta(\sigma) \subseteq \m{O}$ that are legal with respect to~$\sigma$. 
An algorithm $\mathcal{A}$ solves a task $\Pi$ if whenever each of the $n$ processes  is initially given a private input such that the sets of inputs defines a configuration $\sigma\in\m{I}$, any execution of $\mathcal{A}$ results in the (non crashing) processes deciding  output values  forming a configuration $\tau\in \m{O}$ that is legal for~$\sigma$ according to the input-output specification~$\Delta$, i.e., $\tau\in\Delta(\sigma)$.

\paragraph{Executions.}

An \emph{execution} (of an algorithm) is a sequence of \emph{steps}, i.e., a sequence of read or a write operations in the shared memory performed by any of the processes. A crucial point is that, although each of the processes is deterministic, the overall system is non-deterministic, due to asynchrony and crash failures. A same algorithm $\mathcal{A}$ may thus have different executions, due to different interleaving between the read and write operations performed in the shared memory, and due to the fact that some processes may not terminate due to crashes. In fact, some processes may not even take any step in case they crash at the very beginning of an execution.  A \emph{solo} execution is an execution in which only one of the $n$ processes take steps, all the others being either crashed or very slow. A process that does not crash in an execution, is said to be \emph{correct} in that execution. 

\paragraph{Resiliency.}

The model has different variants, parametrized by the maximum number of crashes which may occur during any execution (of an algorithm). For every $t\in\{1,\dots,n-1\}$, the \emph{$t$-resilient} model assumes that at most $t$ processes can crash during any execution. The case $t=n-1$ deserves  special attention, is  referred to as the \emph{wait-free} model. Indeed, if all but one process may crash during any execution, a process~$i$ cannot ``wait'' for another process~$j$ to perform a write in the shared memory (e.g., for $p_j$ to increase a counter in its register $R_j$), because, $p_i$ cannot distinguish between the case where $p_j$ is very slow (due to, e.g.,  contention), and the case where $p_j$ has crashed. Instead, for $t<n-1$, a process can wait for up to $n-(t+1)$ other processes, as for sure $n-t$ processes do not crash. 

\paragraph{Consensus.}

Consensus is arguably the most important task in  distributed computability. In \emph{consensus,} every process $i$ starts with a value $x_i\in I$ taken from some finite set~$I$, and every correct process $i$ must decide an output  value $y_i$ such that (1)~the value $y_i$ is equal to the input value $x_j$ of some $p_j$, and (2)~the output values decided by the correct processes are identical. The instance of consensus in which $I=\{0,1\}$ is \emph{binary consensus}. As mentioned previously, solving consensus is  impossible, even in the stronger model considered in the paper (i.e., the $1$-resilient model), a result originally proved in a message passing system~\cite{FischerLP85}, but can be  extended to shared memory~\cite{attiya2004distributed}.

\begin{lemma} \label{lem:impossibility-of-consensus}
For every $n\geq 2$, binary consensus cannot be solved by $n$ processes in the $1$-resilient shared-memory model.
\end{lemma}

\paragraph{Approximate Agreement.}


Approximate agreement~\cite{DolevLPSW83} is a relaxation of binary consensus, parame\-tri\-zed by a ``degree of agreement''~$\epsilon>0$. When the domain is discrete $\{0,1\}$, for $\epsilon \in (0,1]$, the binary $\epsilon$-agreement task is the following. 
Each process starts with an input value in $\{0,1\}$. It must produce an output value in the interval $[0,1]$,
such that (1)~if all process start with the same value~$x\in \{0,1\}$, then they must decide that value~$x$, and (2)~otherwise 
they can output any values in $[0,1]$ as long all these values are at most $\epsilon$ apart from each other.  
To be formalized as a task, a discretized version of approximate agreement has been introduced (see~\cite{bookHerlihyKR2013}), in which $\epsilon=1/k$ for some integer $k\geq 1$, and it is required that all the output values be of the form $m/k$, for some $m\in\{0,\dots,k\}$.  
As opposed to consensus,  $\epsilon$-agreement is solvable, for every $\epsilon>0$, even in the weakest variant considered in the paper, i.e., the wait-free model.

\begin{lemma}\emph{(\cite{DolevLPSW86,HoestS06})}\label{lem:possibility-of-approx-agreement}
For every $\epsilon>0$ and every $n\geq 2$, $\epsilon$-agreement can be solved by $n$ processes in the wait-free shared-memory model.
\end{lemma}

\paragraph{Size of the Registers.}

This paper is mostly interesting in the \emph{size} of the shared registers $R_1,\dots,R_n$. We are questioning whether there is a bound $b\geq 1$ such that all tasks solvable using registers of unbounded size are solvable using $b$-bit registers. Ideally, we would like $b$ to be independent of the number $n$ of processes. An intermediate case however is the case where $b$ is allowed to grow with $n$, but independently of the tasks. There is a subtlety regarding the role of the registers, because they serve two purposes. One is informing the other processes of the input given to each process, and the other is coordination between the processes so that they can eventually produce outputs. Our impossibility results hold for tasks with constant inputs, actually inputs in $\{0,1\}$. 

For handling tasks with arbitrarily large inputs in the context of our positive results using registers of bounded size (whether it be constant or function of~$n$),  we enhance the model by assuming one additional SWMR \emph{write-once} register $I_i$ per process $i\in\{1,\dots,n\}$. In this case, the first operation of each process $i$ must be writing its input~$x_i$ (if any) in its register~$I_i$, thanks to the atomic instruction $I_i.\mathsf{write}(x_i)$. This register can be read at will by the other processes, thanks to the instruction  $I_i.\mathsf{read}()$, but $p_i$ never writes again to its  register~$I_i$. That is, the registers $I_1,\dots,I_n$ are just to store input values, and cannot be used for coordination between the processes. Coordination must only be based on using the $b$-bit registers $R_1,\dots,R_n$. The input registers $I_1,\dots,I_n$ are required only when the range of inputs exceeds what can be stored in the registers $R_1,\dots,R_n$, e.g., inputs in $\{1,\dots,n\}$ with constant-size registers. Again, our impossibility results hold even for binary inputs. 

\paragraph{Snapshots and Immediate Snapshots.} 

A \emph{snapshot}~\cite{AfekADGMS93} is a mechanism, a.k.a.~\emph{object}, allowing any process $i$ to read all the registers of the other processes as if this would occur atomically, that is, with the guarantee that no process $j$ writes in its register while $p_i$ is performing its snapshot. The instruction $\mathsf{snapshot}()$ returns all the values stored in the registers in the memory, as an $n$-entry vector. An  \emph{immediate} snapshot~\cite{BG92}  is an even more sophisticated object enabling any process $i$ to write in its register, and then snapshot the memory, as a single atomic action, as if the snapshot occurred ``immediately'' after the write. As opposed to all the previously mentioned communication instructions (read, write, and snapshot), several immediate snapshots may occur simultaneously. If this is the case, then every process performing the immediate snapshot gets in its own snapshot all the values written by the other processes concurrently performing the immediate snapshot. W.l.o.g., we can assume access to immediate snapshot, due to the following result. 

\begin{lemma}\emph{(\cite{BG92})}\label{lem:possibility-of-immediate-snapshot}
For every $n\geq 2$, immediate snapshot  can be implemented by $n$ processes in the wait-free read/write shared-memory model.
\end{lemma}

\paragraph{The Iterated Model.} 

The model is defined here assuming processes systematically use immediate snapshots to communicate, and so is called the \emph{iterated immediate snapshot} model, or IIS for short. The IIS model is similar to the shared-memory model, excepted that the processes have access to an infinite sequence of shared memories $M_1,M_2,\dots$. Each memory $M_r$ is composed of $n$ SWMR registers $R_{r,i}$, $i=1,\dots,n$, where $p_i$ is the only process allowed to write into~$R_{r,i}$. The crucial characteristic of the model is that, for every $r\geq 1$, when a process performs its $r$th immediate snapshot, it does so in memory~$M_r$. By definition, the iterated model can be simulated by the non-iterated model (assuming unbounded registers). The interesting fact is that both models are computationally equivalent\footnote{This is not obvious, as illustrated by the example of an execution $\m{E}$ in which $p_1$ is faster than all the other processes, which are however supposed to be just one step behind $p_1$ in their immediate snapshots. In this execution~$\m{E}$, $p_1$  runs ``solo'' in the IIS model in the sense that none of its immediate snapshots ever return a value written by another process. This is in contrast to the non-iterated model, in which, at step $r$ in~$\m{E}$, $p_1$ reads all the values written by the other processes at their step $r-1$.}.   

\begin{lemma}\emph{(\cite{BorowskyG97simple})}\label{lem:power-of-IIS}
For every task $\Pi$ and for every $n\geq 2$, $\Pi$ is solvable by $n$ processes in the wait-free shared-memory model if and only if $\Pi$ is solvable by $n$ processes in the wait-free IIS model.
\end{lemma}


\section{High Level Ideas of the Proofs}
\label{sec:sketch}

This section is aiming at providing the readers with some intuition about the techniques used to established our results. The formal proofs can be found in the  sections further in the paper. 

\subsection{Non Universality when a Majority of Processes May Fail}
\label{subsec:sketch:main:wait-free-not-universal}

The proof of Theorem~\ref{main:wait-free-not-universal} is based on demonstrating that, given $f:\mathbb{N}\to\mathbb{N}$, $n>2$, and $t\in (\frac{n}{2},n-1]$, there exists $\epsilon>0$ sufficiently small such that $\epsilon$-agreement cannot be solved by $n$ processes in the $t$-resilient shared-memory model with registers of at most $f(n)$ bits. This will be sufficient for establishing the theorem since, thanks to Lemma~\ref{lem:possibility-of-approx-agreement}, for every $\epsilon>0$, $\epsilon$-agreement is solvable by $n$ processes in the wait-free model, and therefore in the stronger $t$-resilient model as well, whenever using unbounded-size registers. To see why $\epsilon$-agreement cannot be solved by $n>2$ processes whenever $\epsilon$ is too small, let us consider the simplest case of wait-free computing (i.e., up to $n-1$ processes can crash). 

Intuitively, registers of size at most $s=f(n)$ bits causes the following issue. 
If two processes $1$ and $2$ run solo, and reach agreement, their outputs are at most $\epsilon$ apart from each other. However, a register of size~$s$ bits can take at most $2^s$ different values, and therefore, if $\epsilon\ll \frac1{2^s}$, they cannot accurately store their outputs in the shared registers. As a consequence, if a third process $3$ does not take any step before $p_1$ and $p_2$ terminate, it cannot determine its output with sufficient precision for guaranteeing that it will be at most $\epsilon$ apart from the outputs of $p_1$ and~$p_2$. 

Slightly more formally, let us assume that, for every $\epsilon>0$, there exists a wait free algorithm $A_\epsilon$ solving $\epsilon$-agreement for $n\geq 3$ processes, using registers of size $s$ bits. We concentrate on the behavior of $A_\epsilon$ restricted to the executions in which $p_1$ has input~0, and  $p_2$ has  input~1. Note that $p_1$ running solo must output~0, and $p_2$ running solo must output~1. Let $G=(V,E)$ be the graph in which every vertex is a pair $(i,\sigma)$ where $i\in\{1,2\}$ is a process identifier, and $\sigma$ is an final internal state of process~$i$ in some execution of~$A_\epsilon$, i.e., a state in which $p_i$ outputs and terminates. There is an edge between two vertices $(i,\sigma_i)$ and $(j,\sigma_j)$ of~$G$ if $i\neq j$ and there is an execution in which processes $i$ and $j$ terminate in respective internal state $\sigma_i$ and $\sigma_j$. Note that $G$ includes a vertex $v_1=(1,\sigma_1)$ corresponding to a solo execution of~$p_1$, as well as a vertex $v_2=(2,\sigma_2)$ corresponding to a solo execution of~$p_2$. To each vertex $(p,\sigma)$ of $G$ corresponds  the output value returned by process~$p$ in (final) state~$\sigma$. The value associated to $v_1$ is~0, and the value associated to $v_2$ is~1. A crucial observation is that $G$  must include a path between the vertex~$v_1$ and the vertex $v_2$. Indeed, if these two vertices were disconnected in~$G$, then $p_1$ and $p_2$ could solve consensus, by having all vertices of the connected component of $G$ containing~$v_1$ outputting~0, and  all vertices of the connected component of $G$ containing~$v_2$ outputting~1. Since consensus is not solvable wait-free (cf. Lemma~\ref{lem:impossibility-of-consensus}), $v_1$ and $v_2$ must be connected by a path $P$ in $G$ (see, e.g.,~\cite{BMZ90charac} for more details about such an argument). 

Let $P=(w_0,w_1,\dots,w_{2k+1})$ with $w_0=v_1$ and $w_{2k+1}=v_2$. Let us associate to each vertex $w_i$ of $P$ the output value~$y_i\in[0,1]$ produced by~$A_\epsilon$, for $i=0,\dots,2k+1$.  By definition of $\epsilon$-agreement, we must have $|y_i-y_{i+1}|\leq \epsilon$. In particular, this implies that the path $P$ must be of length at least $1/ \epsilon$, because its two extremities, $w_0=v_1$ and $w_{2k+1}=v_2$, output~$y_0=0$ and $y_{2k+1}=1$, respectively. That is, $2k+1\geq 1/\epsilon$.
For every $i\in\{0,\dots,k\}$, let $\ell_i$ be the  largest index such that $y_{\ell_i}= 2i\epsilon$. We focus on the $k+1$ pairs of vertices $(w_{\ell_i},w_{\ell_i+1})$ for $i=0,\dots,k$. Note that $\ell_0<\ell_1<\dots<\ell_k$, and, for every $i\in\{0,\dots,k\}$,  $y_{\ell_i+1}=(2i+1)\epsilon$, which in particular implies that $w_{\ell_i+1}\neq w_{\ell_{i+1}}$. So, overall, we have that,  for $i=0,\dots,k$, 
$
(y_{\ell_i},\; y_{\ell_i+1})=\big(2i\epsilon,\; (2i+1)\epsilon\big).
$
Once $p_1$ and $p_2$ have terminated, let $p_3$ start executing $A_\epsilon$. The output of $p_3$ will depend only on of  its own input value, and on the values that it reads in the registers $R_1$ and $R_2$ of $p_1$ and~$p_2$. 

Since the registers have size $s$ bits, there are at most $2^{2s}$ different values for the content of the pair $(R_1,R_2)$. If 
$1/\epsilon > 2^{2s+1}$, then, by the pigeon hole principle, there are two distinct pairs of vertices $(w_{\ell_i},w_{\ell_i+1})$ and $(w_{\ell_j},w_{\ell_j+1})$ corresponding to a same content for the pair of registers $(R_1,R_2)$. However, $p_1$ and $p_2$ output $(2i\epsilon,\; (2i+1)\epsilon)$ for $(w_{\ell_i},w_{\ell_i+1})$, whereas $p_1$ and $p_2$ output $(2j\epsilon,\; (2j+1)\epsilon)$ for $(w_{\ell_j},w_{\ell_j+1})$. Since $i\neq j$, there are thus four different outputs for $p_1$ and $p_2$ for which $p_3$ receive the same information stored in the registers $R_1$ and $R_2$ of these two processes. It follows that, whatever is the output for $p_3$, this output is at least $2\epsilon$ apart from one of these four values. As a consequence, there is an execution for which $A_\epsilon$ does not solve $\epsilon$-agreement, a contradiction. 

Note that this sketch of proof heavily uses the wait-free assumption. For instance, if at most $n-2$ processes could fail instead of at most $n-1$ as in the wait-free case, then, in the scenario depicted in the proof, the two processes $1$ and $2$ could wait for each other, and coordinate so that to solve consensus. Instead, in the wait-free setting, these processes are bounded to output values that are $\epsilon$ apart from each other in some executions, precisely because none of the two processes can ever wait for the other process, at any point in time during all executions.  The formal proof, for the wait-free model as well as for all $t$-resilient models, $\frac{n}2< t <n$, can be found in Section~\ref{sec:no_dum}.

\subsection{Constant-Size Registers are Universal for 2-Process Systems}

To establish Theorem~\ref{theo:case-of-two-proc}, that is, to show that any task solvable by 2~processes in the wait-free shared-memory model with unbounded-size registers can be solved  by 2~processes  in the wait-free shared-memory model with 1-bit registers, we proceed in two steps. 

First, we show that, for every $\epsilon>0$, $n=2$ processes can solve $\epsilon$-agreement with 1-bit registers. In this algorithm, each of the two processes alternates writing 0 and 1 in its register, and reading the value in the other register. It stops whenever it reads twice the same value in the register of the other process. Intuitively, this allows the two processes to check whether they are running concurrently, and to stop whenever they become desynchronized, at which point they can decide output values that are close together. Which value each of the two processes adopt depends on the parity of the round at which they become desynchronized, and on which of the two processes went faster than the other. Nevertheless,  the fact that the two processes are not desynchronized by more than one round allows them to adopt  values close to each other. The number of rounds is limited up to $\Theta(1/\epsilon)$ read/write operations, which is enough for guaranteeing a precision  of $\epsilon$ in case the two processes run concurrently in locksteps. 

Second, we start from the known fact that a black box solving $\epsilon$-agreement can be used for solving any wait-free solvable task, in 2-process systems~\cite{bookHerlihyKR2013}. In a nutshell, we show that this holds even when using 1-bit registers. To get the intuition of why this is true, we use the classical framework introduced by Biran, Moran and Zaks~\cite{BMZ90charac}, who gave necessary
and sufficient conditions for tasks to be solvable in the $1$-resilient model (which coincides with the wait-free model for $n=2$ processes). 
Roughly, we use the same type of graphs as the one used in Section~\ref{subsec:sketch:main:wait-free-not-universal}. That is, 
at any point in time of the execution of a protocol, all possible states of the system can be described by a graph whose vertices are pairs $(p,\sigma)$ where $p\in\{1,2\}$ is a process, and $\sigma$ is a possible local state of $p$ at the considered time. Two vertices are linked by an edge if the two vertices correspond to two different processes in compatible states. Starting from the input graph $G_0$ including all vertices $(p,x)$ where $x$ is a possible input value for process $p$, the computation evolves as a sequence of graphs $(G_t)_{t\geq 0}$, where $G_t$ displays all possible states of the system after $t$ steps. Assuming that all processes terminate at the same step~$T$, an algorithm specifies an output values for each vertex of $G_T$, i.e., for each possible state $\sigma$ of a process~$p$ after $T$ steps. Now, all possible output states also form a graph $H$ including all vertices $(p,y)$ where $y$ is a possible output value for process $p$. So the algorithm maps every vertex $(p,\sigma)$ of $G_T$ to a vertex $(p,y)$ of $H$. Moreover, this mapping must be a morphism, i.e., it must preserve the edges, as any possible global states after $T$ steps, that is, any possible edge $\{(1,\sigma_1),(2,\sigma_2)\}$ of $G_T$, must be mapped to a legal global output states, i.e., to an edge $\{(1,y_1),(2,y_2)\}$ of $H$. 

The structure of the graph $G_T$ depends on the communication model of the system. In the case of 1-resilient computing, this graph can be viewed as a form of refinement of the input graph $G_0$, mostly obtained by subdividing the edges of~$G_0$, which intuitively corresponds to capturing all possible interleavings of the read and write operations performed by the two processes (see~\cite{BMZ90charac}). Each process $p$ knows in which vertex $(p,\sigma)$ of $G_T$ it stands after it has performed $T$ steps, starting from its input $(p,x)$ and repeating $T$ times (1)~writing  its entire current history in memory, and (2)~reading the current history of the other process. Here comes the use of approximate agreement. An alternative approach is to let the two processes solving $\epsilon$-agreement, for some suited $\epsilon$ depending on $T$, and  small enough so that there is a mapping from the output graph of $\epsilon$-agreement to $G_T$. That is, given an input edge $\{(1,x_1),(2,x_2)\}$ of $G_0$, the two processes can $\epsilon$-agree on an edge $\{(1,\sigma_1),(2,\sigma_2)\}$ of $G_T$ without exchanging histories, and actually using 1-bit registers only. Once they have agreed on $\{(1,\sigma_1),(2,\sigma_2)\}$, the two processes can map the edge to a legal output configuration. 
The formal proof can be found in Section~\ref{sec:equality}.

\paragraph{Remark.}

Our algorithm  for solving $\epsilon$-agreement with 1-bit registers used to establish Theorem~\ref{theo:case-of-two-proc} has step-complexity equal to $O(\frac1\epsilon)$, which is exponentially slower than the fastest algorithm solving $\epsilon$-agreement with unbounded registers --- the step complexity of $\epsilon$-agreement is $\Theta(\log\frac1\epsilon)$ (see, e.g.,~\cite{FuggerNS21}). In Theorem~\ref{thm:lcis} (cf. Section~\ref{sec:faster-eps-agr}), we show that an algorithm solving $\epsilon$-agreement with optimal step-complexity  $O(\log\frac1\epsilon)$ can be designed for a shared-memory with constant-size registers, that is, registers on $O(1)$ bits instead of just 1~bit.  
 
\subsection{Universality when a Minority of Processes May Fail}

To establish Theorem~\ref{theo:small-values-of-t}, given an algorithm $\m{A}$ solving a task $\Pi$ in the $t$-resilient shared-memory model with unbounded registers, one aims at constructing an algorithm $\m{B}$ solving $\Pi$ with registers of size $O(t)$ bits. This construction proceeds in three phases. 

\begin{enumerate}
\item Since $t<\frac{n}{2}$, the existence of  $\m{A}$ implies~\cite{AttiyaBD95} that there exists an $n$-process \emph{message-passing} algorithm $\m{A}'$ solving $\Pi$ when at most $t$ processes may crash. Recall that, in the message-passing model, instead of communicating via a shared memory, the processes communicate by exchanging messages. Every process can send a message to every other process, hence the communication network is complete. Communication channels are FIFO, and no messages are lost. However, the time required to deliver a message is unbounded. 

\item We show that the $t$-resilient message-passing model (which assumes a complete network) can be simulated by a network with fewer communication links between processes,  as long as the network is $(t+1)$-connected, i.e., removing at most $t$ nodes of the network keeps it connected. This results in an algorithm  $\m{A}''$.  

\item We rely on the  \emph{alternating bit} protocol (see~\cite{bartlett1969note,lynch1968computer}), to simulate algorithms for the $n$-node $(t+1)$-connected network in the shared memory model with registers of size $O(t)$ bits. Applied to algorithm $\m{A}''$, this results in the desired algorithm~$\m{B}$.
\end{enumerate}

\noindent More details can be found in Section~\ref{sec:minority}.

\subsection{Universality of the Iterated Immediate Snapshot (IIS) Model}
 
For proving Theorem~\ref{theo:iterated-immediate-snapshot}, we show that,  for any $n\geq 2$, any task solvable wait-free by $n$ processes in the IIS model with unbounded registers is solvable  wait-free  in the  IIS model with 1-bit registers (at each level of the memory). To establish this result, we show that \emph{full-information} protocols, i.e., protocols using unbounded registers in which processes can write all the information acquired so far, can be emulated in the  IIS model using $1$-bit registers. We actually consider a weaker variant of the IIS model, where the immediate snapshots are merely replaced by  \emph{collect} operations. In the iterated collect model (IC), every memory $M_r$, $r\geq 1$, is read by each process using a $\mathsf{collect}()$ operation~\cite{attiya2004distributed}, which merely consists of reading the $n$ registers of $M_r$ one by one, in an arbitrary order. We then proceed in two steps: 

\begin{enumerate}
\item We show that any full-information protocol in the IIS model with unbounded size registers can be simulated in the IC model with unbounded size registers. For this, we simulate one round of the IIS model with $n$ $\mathsf{write}$/$\mathsf{collect}$ iterations using the  (immediate) snapshot algorithm presented by Borowsky and Gafni in~\cite{BG92}. 

\item We show that any full-information protocol solving a task in the IC model with unbounded-size registers can be simulated in the IIS model using $1$-bit registers.  In a sense an iterated model provides a counter to the processes ``for free'', corresponding to the index~$r$ of the memory $M_r$. We show that, with this implicit counter, 1-bit registers suffice for the model to be universal. Intuitively, instead of writing a large value~$v$ in memory, a process can merely write~1 in its register in~$M_v$, and write~0 in all the other registers. This approach is over simplified, and does not suffice for  simulating the IIS model using $1$-bit registers. Neverheless, this provides hint on how one can play with the level~$r$ of the memory~$M_r$ to encode large value. 
\end{enumerate}

\noindent More details can be found in Section~\ref{sec:iterated}. The rest of the paper is dedicated to the formal proofs of the arguments sketched in this section.


\section{Non Universality when a Majority of Processes May Fail}
\label{sec:no_dum}

In this section, we formally prove Theorem~\ref{main:wait-free-not-universal}, that is, we show that when a majority of the processes may fail, restricting the size of the registers strictly limits the computational power of the shared memory model. Recall that, for every $n\geq 2$, and for every $\epsilon>0$, $\epsilon$-agreement is solvable wait-free by $n$ processes using unbounded registers (cf. Lemma~\ref{lem:possibility-of-approx-agreement}). The proposition below  implies Theorem~\ref{main:wait-free-not-universal}, keeping in mind that, for every $\epsilon>0$, since $\epsilon$-agreement is solvable wait-free, it is solvable under the stronger $t$-resilient model, for every $t\in\{1,\dots,n-1\}$, using registers of unbounded size.  

\begin{proposition}
\label{prop:no_dum}
For any function $f:\mathbb{N}\to\mathbb{N}$, for any integer $n>2$, and for any integer $t$ with $\frac{n}{2} < t < n$, there exists $\epsilon>0$ such that $\epsilon$-agreement cannot be solved by $n$ processes in the $t$-resilient shared-memory model with registers of at most $f(n)$ bits. 
\end{proposition}

\begin{proof}
 Again, the proof is by contradiction.  Let us assume that, for every $\epsilon>0$, there exists a $t$-resilient algorithm $A_\epsilon$  solving $\epsilon$-agreement for $n<2t$ processes, using registers of size $f(n)$ bits. It is sufficient to consider  $\epsilon=1/k$ for some odd integer $k>1$. For such an $\epsilon$, the sets of possible inputs and outputs are respectively  $\m{I} = \{0,1\}$ and $\m{O} = \{0,\frac{1}{k},\ldots,\frac{k-1}{k},1\}$. For every $\ell \in \{0,\dots,k-1\}$, let 
 \[
 O_\ell = \left\{\frac{\ell}{k}, \frac{\ell+1}{k}\right\}.
 \]
That is, $O_\ell$ is a valid legal set of outputs in an execution in which the input set is $\{0,1\}$. 

\begin{claim}\label{clm:Ol}
For every $\ell\in\{0,\dots,k-1\}$, there exists an execution of $\m{A}_\epsilon$ satisfying that 
(1)~only  processes $1,\ldots, n-t+1$ take steps, and (2)~at least $n-t$ processes decide, and the set of outputs is equal to~$O_\ell$. 
\end{claim}

To establish that claim, let us assume, for the purpose of contradiction, that  there exists $\ell\in\{0,\dots,{k-1}\}$ such that, for every  execution of $\m{A}_\epsilon$ in which only the processes $1,\ldots, n-t+1$ take steps, the set of outputs is not~$O_\ell$.  To establish a contradiction, we are going to show that, under this assumption, processes $1,\ldots,n-t+1$ could use algorithm  $\m{A}_\epsilon$ to solve $1$-resilient binary  consensus, contradicting Lemma~\ref{lem:impossibility-of-consensus}. Note that since $n\geq 3$, and $t>n/2$, we have always $t>1$, and thus $n-t+1\leq n-1$.

Let us assume that every process $i \in \{1,\ldots,n-t+1\}$ is given an input for binary consensus (i.e., $0$ or $1$). Process~$i$ first runs algorithm $\m{A}_\epsilon$. Since we are interested in solving binary consensus in the $1$-resilient model, we consider only the executions in which at most one process in $\{1,\ldots,n-t+1\}$ fails. Hence, in the corresponding execution of $\m{A}_\epsilon$, we have that $t-1$ processes fail initially, and at most  1~additional process fails, perhaps after having taken some steps. Since $\m{A}_\epsilon$ tolerates $t$ failures, each process  $i \in \{1,\ldots,n-t+1\}$ that does not fail will eventually produce an output, say~$d_i$. To solve binary consensus, process~$i$ outputs
\[
\mbox{out}(p_i)=\left\{\begin{array}{ll}
 0 & \mbox{if $d_i \leq \frac{\ell}{k}$} \\
 1 & \mbox{otherwise.}
\end{array}\right.
\]
In this way, binary consensus is solved. Indeed, if the only input is~$0$ (resp.,~$1$), then the only value returns by $\m{A}_\epsilon$ is~$0$ (resp.,~$1$), by the validity property of $\epsilon$-agreement. Moreover, let $D$ be the set of output values returned by $\m{A}_\epsilon$. As $\m{A}_\epsilon$ solves $\epsilon$-agreement, we have 
\[
D \subseteq O_{\ell'} = \left\{\frac{\ell'}{k},\frac{\ell'+1}{k}\right\}
\] 
for some $\ell'\in\{0, \dots, k-1\}\smallsetminus \{\ell\}$. If $D$ is a singleton, then all process that decide have outputted the same value, as desired. If $D=O_{\ell'}$, then, as $\ell' \neq \ell$, we have 
\[
\frac{\ell'}{k} < \frac{\ell'+1}{k} \leq \frac{\ell}{k} \;\;\mbox{or}\;\;  \frac{\ell+1}{k} \leq \frac{\ell'}{k} < \frac{\ell'+1}{k}.
\]
Every process that decides thus output~$0$ in the first  case, and $1$ in the second case, from which agreement follows. This contradicts the impossibility of $1$-resilient binary consensus (cf. Lemma~\ref{lem:impossibility-of-consensus}), and completes the proof of Claim~\ref{clm:Ol}.

\bigbreak

We now set 
\[
k = 2\left(2^{f(n)}\right)^{n-t+1} + 1.
\]
For each $\ell \in \{0,\ldots, k-1\}$, let $\m{E}_\ell$ be a finite execution of $\m{A}_\epsilon$ in which only processes in $\{1,\ldots, n-t+1\}$ take steps, and the set of outputs is~$O_\ell$. Such an execution exists by Claim~\ref{clm:Ol}. Let $w_\ell$ be the binary word formed by concatenating the contents of the registers $R_1,\ldots,R_{n-t+1}$ immediately after each process $1,\ldots, n-t+1$ has decided. Since each register has size at most $f(n)$ bits, the set $W$ of possible such words satisfies 
\[
  |W| = \left(2^{f(n)}\right)^{n-t+1}. 
\]

Let us consider the set of $\frac{k-1}{2} +1$ output sets $\{O_0,O_2,O_4,\ldots,O_{k-1}\}$.
These sets are mutually exclusive in the sense that if, in some execution of $\m{A}_\epsilon$, the set of outputs of some processes is exactly $O_{2i}$, for some $i\in\{0,\dots, \frac{k-1}{2}\}$, then no other processes output a value $d \in O_{2j}$ for any $j \neq i$ in that execution. Indeed, assume that in some execution $\m{E}$,  $d \in O_{2j}$ along with the values 
\[
d_{2i} = \frac{2i}{k}, \;\mbox{and}\; d_{2i+1} = \frac{2i+1}{k} 
\]
are output. If $j< i$, we have 
\[
d_{2i+1} - d = \frac{2i+1}{k} - d \geq \frac{(2i)+1 - (2(i-1)+1)}{k} = \frac{2}{k} > \epsilon.
\]
Similarly, if $i< j$, then $d-d_{2i} \geq \frac{2}{k}$. Hence, in both case $\epsilon$-agreement is not satisfied in $\m{E}$.  

Thanks to the setting of~$k$, and by the pigeon hole principle, we get that there exists $i\neq j$ in $\{0,\dots,\frac{k-1}{2}\}$ such that $w_{2i} = w_{2j}$. Let $\m{E}'_{2i}=\m{E}_{2i}\circ \m{E}$ be an extension of $\m{E}_{2i}$ in which, immediately after $\m{E}_{2i}$, all processes $1,\ldots,n-t+1$ crash, and all the $t-1$ processes $n-t+2,\ldots,n$  take steps, say in a round robin fashion), until  one of them  decides  some value~$d$. Note that  $d \in O_{2i}$ since both values $\frac{2i}{k}$ and $\frac{2i+1}{k}$ have been decided in $\m{E}_{2i}$.  The inputs of processes $\{n-t+2,\ldots,n\}$ are defined arbitrarily. As $t > \frac{n}{2}$, we have $t \geq n-t+1$, and thus $\m{E}'_{2i}$ is an execution with at most $t$ failures. As $\m{A}_\epsilon$ tolerates up to $t$ failures, at least one of the processes $n-t+2,\ldots,n$ must decide after taking finitely many steps. Since, in $\m{E}_{2i}$ as well as in $\m{E}_{2j}$, no processes in $\{n-t+2,\ldots,n\}$ take step, and the state of the shared registers is the same at the end of both executions, the executions $\m{E}_{2i}$ and $\m{E}_{2j}$ are indistinguishable for the processes $n-t+2,\ldots,n$. The execution $\m{E}'_{2j} = \m{E}_{2j} \circ \m{E}$ is thus a valid execution of $\m{A}_\epsilon$, indistinguishable from $\m{E}'_{2i}$ for processes in $\{n-t+2,\ldots,n\}$. It follows that $d$ is also the output value in $\m{E}'_{2j}$ for some process. Therefore, $\epsilon$-agreement is violated in $\m{E}'_{2j}$ since $d \notin O_{2j} = \{\frac{2j}{k},\frac{2j+1}{k}\}$ and both $\frac{2j}{k}$ and $\frac{2j+1}{k}$ are outputted in the prefix $\m{E}_{2j}$ of $\m{E}'_{2j}$, contradicting the correctness of $\m{A}_\epsilon$. 
\end{proof}


\section{Constant-Size Registers are Universal for 2-Process Systems}
\label{sec:equality}

We present a  proof of Theorem~\ref{theo:case-of-two-proc} by providing a wait-free $\epsilon$-agreement protocols for $n=2$ processes.  Then we show how to use the approximate agreement protocol for solving any task that is solvable wait-free using registers of unbounded size.


\subsection{Wait-Free $\epsilon$-Agreement with 1-bit Registers}
\label{sec:epsilon_agreement}

Given any $\epsilon\in (0,1)$, let $k$ be an integer such that $\frac{1}{2k+1} \leq \epsilon$.   
We will show that the protocol ${\cal A}_{k}$ described in Algorithm~\ref{alg:epsilon} solves $\epsilon$-agreement for two processes $p_1$ and $p_2$. The first step in the code is for each process  $i\in\{1,2\}$ to write its input in the write-once special register  $I_i$.
The shared registers $R_1$ and $R_2$ are  initialized to~$0$, and registers $I_1$ and $I_2$ to $\bot$. 

\newcommand{\vprec}{\mathit{prec}}
\newcommand{\new}{\mathit{new}}
\newcommand{\who}{\mathit{who}}
\newcommand{\vme}{\mathit{me}}
\newcommand{\vother}{\mathit{other}}
\begin{algorithm}[tb]
  \caption{Approximate agreement protocol $\m{A}_k$ for two processes. Code for process $i \in \{1,2\}$ with input $\mathit{myInput} \in \{0,1\}$}
  \label{alg:epsilon}
  \begin{algorithmic}[1]
    \Statex \textbf{local variables:} 
    \Statex \hspace{\algorithmicindent} $\vprec,\new,x_1,x_2 \in \{0,1\}$ \Comment{$\vprec$ initialized to 0}
    \Statex \hspace{\algorithmicindent} $\vme,\vother, \who\in \{1,2\}$ 
    \Comment{$\vme$ and $\vother$ resp.~initialized to $i$ and $j\in\{1,2\}\smallsetminus \{i\}$ at $p_i$}
    \Statex \hspace{\algorithmicindent} $r\in \{1,\dots,k\}$
    
    \Begin
    \State $I_{me}.\mathsf{write}(\mathit{myInput})$\label{line:writeinput}
    \For{$r=1,\ldots,k$} \label{line:forstart}
    		\State $R_{\vme}.\mathsf{write}(r \bmod 2)$
    		\State $\new \gets R_{\vother}.\mathsf{read}()$
    		\If{$\new \neq \vprec$} 
			$\vprec \gets \new$
   		 \Else{}\label{line:break} 
		 	\textbf{break} \Comment{leave the for-loop if the same value is read twice}
    		\EndIf
    \EndFor
    \State $x_{me} \gets I_{me}.\mathsf{read}()$
    \State $x_{other} \gets I_{other}.\mathsf{read}()$
    \If{($x_{\vother} = \bot$ \textbf{or}  $x_{me} = x_{other}$)}\label{line:decIdem}\label{line:decinit} \textsl{return} $x_{me}$ 
    \Else
    		\If{($r=k$  \textbf{and} $\new = k \mod 2$)} \Comment{$\vme$ left the for-loop after  $k$ iterations}
      			\If{$r \bmod 2 =0$}\label{line:decpairk} 
                        $\who\gets \vme$ 
                        \textbf{else} \label{line:decimpairk}$\who\gets \vother$ 
			\EndIf
	      		\State \label{line:dec_r=k}\textsl{return}\big($(x_\who+k)/(2k+1)$\big)
		\Else{} \Comment{ $\vme$ left the for-loop because it has read the same value twice}
      			\If{$r \bmod 2 =0$}\label{line:decpair} 
				$\who\gets \vother$ 
			\textbf{else}\label{line:decimpair} 
				$\who\gets \vme$ 
			\EndIf
 			\State\label{line:dec_r<k}\textsl{return}\big($x_\who + ((-1)^{x_\who}(r-1))/(2k+1)$\big)
  		 \EndIf
     \EndIf
\End
\end{algorithmic}
\end{algorithm}

Each process $i$ repeatedly writes (at most $k$ times) in its own register $R_{i}$,  and reads the register of the other process.  The processes  alternate writing $0$ and~$1$. If a process reads twice the same value, then the process stops reading and writing, and computes its output. A possible execution of the algorithm is given in Figure~\ref{fig:1bit_eps_exe}.

\begin{figure}[htbp]
  \centering
\newcommand{\level}{1}
\newcommand{\scale}{1.8}
\newcommand{\drawpoint}[2]{
    \ifthenelse{\isodd{#1}}{ 
      \tkzDrawPoint[Bstyle](B#2#1)
    }{ 
      \tkzDrawPoint[Astyle](A#2#1)
    }
  }
\newcommand{\drawsoloedge}[2]{
  \pgfmathsetmacro\childsolo{int(3*#1)}
  \pgfmathsetmacro\nextlevel{int(#2+1)}
  \ifthenelse{\isodd{#1}}{ 
    \draw[SBedge] (B#2#1) -- (B\nextlevel\childsolo);
  }
  { 
    \draw[SAedge] (A#2#1) -- (A\nextlevel\childsolo);
  }
}

\newcommand{\drawsoloedgesandpoints}[3]{
  \drawsoloedges{#1}{#2}{#3}
  \pgfmathsetmacro\nextlevel{int(#3+1)}
  \pgfmathsetmacro\firstindex{int(#1*3)}
  \pgfmathsetmacro\lastindex{int(#2*3)}
  \drawpoints{\firstindex}{\lastindex}{\nextlevel}
  }
\newcommand{\drawsoloedges}[3]{
  \pgfmathsetmacro\secondindex{int(#1+1)}
  \ifthenelse{#1=\secondindex}{
    \drawsoloedge{#1}{#3}
    \drawsoloedge{\secondindex}{#3}
  }
  {
    \foreach \i in {#1,\secondindex,...,#2} {
      \drawsoloedge{\i}{#3}
    }
  }
  }
\newcommand{\drawpoints}[3]{
    \pgfmathsetmacro\secondindex{int(#1+1)}
    \ifthenelse{#1=\secondindex}{
      \drawpoint{#1}{#3}
      \drawpoint{\secondindex}{#3}
    }
    {
      \foreach \i in {#1,\secondindex,...,#2} {
        \drawpoint{\i}{#3}
      }
    }
}

\hspace*{-2.7cm}\begin{tikzpicture}[
   Astyle/.style = {color=black, fill= black, shape=circle, minimum size=5},
   Bstyle/.style = {color=red,  fill= red, shape=circle,minimum size=5},
   Gstyle/.style = {color=gray, fill= gray, shape=circle,minimum size=5},
   Aedge/.style = {dashed, black, thick},
   Bedge/.style = {dashed, red, thick},
   SAedge/.style = {Aedge},
   SBedge/.style = {Bedge}]

   \foreach \r in {0,1,2,3,4}{
     \pgfmathsetmacro\y{-\r}
     \tkzDefPoints{-5.4/\y/D\r, 5.4/\y/F\r}
     \tkzDrawLine[color=gray,thin,dashed](D\r,F\r)
     \tkzLabelLine[gray,above right,pos=1.1](D\r,F\r){\small{round $\r$}}
     \pgfmathsetmacro\xA{-(1*\scale^\r)}
     \pgfmathsetmacro\xB{1*\scale^\r}
     \tkzDefPoints{\xA/\y/A\r, \xB/\y/B\r}

     \ifthenelse{\r=0}{ 
   
     }
     {
       \tkzCalcLength(A\r,B\r)\tkzGetLength{ABl} 
       \renewcommand{\level}{\r}
       \pgfmathsetmacro\nx{3^\level}
       \foreach \x in {0,2,...,\nx}{
         \pgfmathsetmacro\mypos{\x/(\nx)}
         \tkzDefPointOnLine[pos=\mypos](A\r,B\r)
         \tkzGetPoint{A\r\x}
       }
       \foreach \x in {1,3,...,\nx}{
         \pgfmathsetmacro\mypos{\x/\nx}
         \tkzDefPointOnLine[pos=\mypos](A\r,B\r)
         \tkzGetPoint{B\r\x}
       }
     }
   }
   \tkzDefShiftPoint[A1](0,1){A0ext}
   \tkzDefShiftPoint[B1](0,1){B0ext}
   \draw (A0ext) -- (B0ext);
   \draw (A1) -- (B1);
   \draw (B23) -- (A26);
   \draw (A312) -- (B315);
   \draw (B439) -- (A442);

   \draw[SAedge] (A0ext) -- (A12);
   \draw[SAedge] (A12)   -- (A24);
   \draw[SAedge] (A12)   -- (A24);
   \draw[SAedge] (A24)   -- (A314);
   \draw[SAedge] (A314)  -- (A442);

   \draw[SBedge] (B0ext) -- (B13);
   \draw[SBedge] (B13)   -- (B25);
   \draw[SBedge] (B25)   -- (B315);

   \tkzDrawPoint[Astyle](A0ext)
   \tkzDrawPoint[Bstyle](B0ext)
   \drawpoints{0}{3}{1}
   \drawpoints{3}{6}{2}
   \drawpoints{12}{15}{3}
   \drawpoints{39}{42}{4}

   \tkzLabelPoint[above](A0ext){input:0}
   \tkzLabelPoint[above](B0ext){input:1}
   \tkzLabelPoint[above](A10){\small{10}}
   \tkzLabelPoint[below,gray](A10){\small{0}}
   \tkzLabelPoint[above](B11){\small{11}}
   \tkzLabelPoint[above](A12){\small{11}}
   \tkzLabelPoint[above](B13){\small{01}}
   \tkzLabelPoint[below,gray](B13){\small{1}}

   \tkzLabelPoint[above](B23){\small{01}}
   \tkzLabelPoint[below,gray](B23){\small{$\frac{1}{9}$}}
   \tkzLabelPoint[above](A24){\small{00}}
   \tkzLabelPoint[above](B25){\small{00}}
   \tkzLabelPoint[above](A26){\small{10}}
   \tkzLabelPoint[below,gray](A26){\small{$\frac{8}{9}$}}

   \tkzLabelPoint[above](A312){\tiny{10}}
   \tkzLabelPoint[below,gray](A312){\tiny{$\frac{2}{9}$}}
   \tkzLabelPoint[above](B313){\tiny{11}}
   \tkzLabelPoint[above](A314){\tiny{11}}
   \tkzLabelPoint[above](B315){\tiny{01}}
   \tkzLabelPoint[below,gray](B315){\tiny{$\frac{7}{9}$}}
   
   \tkzLabelPoint[below,gray](B439){\tiny{$\frac{3}{9}$}}
   \tkzLabelPoint[below,gray](A440){\tiny{$\frac{4}{9}$}}
   \tkzLabelPoint[below,gray](B441){\tiny{$\frac{5}{9}$}}
   \tkzLabelPoint[below,gray](A442){\tiny{$\frac{6}{9}$}}
   

 \end{tikzpicture}
 \caption{\sl Executions of Algorithm~\ref{alg:epsilon} with inputs $0$ and $1$ and $k=4$. Process $1$ is represented by a black dot, process $2$ by a read dot. Each dot represents a local state of a process in some execution of the algorithm. The label above a dot represents the state of the memory (the values of the registers $R_1$ and $R_2$) is the corresponding local state. In the execution pictured with dashed line, process $2$ (red) decides after exiting the for loop at iteration 3. Process 1 (black) decides after finishing its for loop.}
 \label{fig:1bit_eps_exe}
\end{figure}
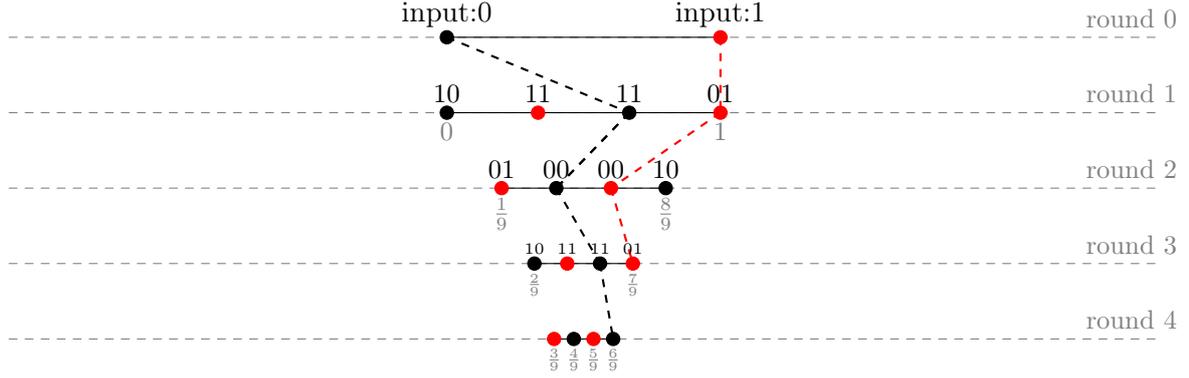

By the code, the algorithm is wait-free as each process performs at most $2k+3$ $\mathsf{read}()$ and $\mathsf{write}()$ operations. In any execution in which only one process decides (by returning a value at line~\ref{line:decIdem}, line~\ref{line:dec_r<k} or line~\ref{line:dec_r=k}), approximate agreement is trivially satisfied. In the following, we consider an  execution $\m{E}$ of the algorithm in which both processes $1$ and $2$ decide. We denote by $r_i$ the value of the local variable $r$ of process $i$ when it decides and by $y_i$ the decision of process $i$. We show that $|y_1 - y_2| \leq \frac{1}{2k+1}$. 

We first show that each process performs the same number of iterations of the for loop or at most one more iteration than the other process. 
\begin{lemma}
  \label{lem:delta_r}
  $|r_1 - r_2| \leq 1$ 
\end{lemma}
\begin{proof}
As the condition for exiting the for loop does not depend on the process' identity, let us assume without loss of generality that $r_1 < r_2$. Process $2$ exists the for-loop (either because $r_2=k$ or at line~\ref{line:break}) after performing $r_2$ iterations. Since process $2$  reaches iteration $r_2$, it must have seen the value in the register $R_{2}$ alternates between $0$ and $1$ $r_2$ times. That is,  the sequence of values returned by the  reads of $R_{1}$ by process $2$ in iteration $1,\ldots,r_2-1$ is $(1,0,\ldots, (r_2-1) \mod 2)$. Therefore, as the initial value of $R_{1}$ is $0$, process $1$ must have written to $R_{1}$ at least $(r_2-1)$ times. By the code, process $1$ writes once to $R_{1}$ in each iteration of its for loop. It follows that process $1$ reaches round $r_2-1$. Hence, $r_1 \geq r_2-1$ which concludes the proof of the Lemma as we assume that $r_1 < r_2$.
\end{proof}

A process exits its for loop prematurely at line~\ref{line:break} when the bit read in the register of the other process is not the parity of its current iteration number. Next Lemma establishes that the two processes cannot exit their for-loop prematurely in the same iteration.
\begin{lemma}
  \label{lem:break}
  Suppose that process $i$ exits its for loop at line~\ref{line:break} in iteration $r_i$ (this is always the case if $r_i < k$.). If  process $3-i$ reaches iteration $r_i$, it does not exits its for loop in this iteration at line~\ref{line:break}. 
\end{lemma}
\begin{proof}
  Suppose that process $1$ exits its for loop prematurely at  line~\ref{line:break} in iteration $r_1$ (The proof is the same for process $2$.). Assume for contradiction that process $2$ also exits its for loop at  line~\ref{line:break} in iteration $r_1$. Note that we thus have $r_1 = r_2$. Let $\rho = r_1 = r_2$ 

  Without loss of generality, let us suppose that in iteration $\rho$, process $1$ writes before process~$2$. That is,  the $\rho$th write to $R_{1}$ by process $1$ occurs before the $\rho$th write to $R_{2}$ by process $2$. Since process~$1$ exits the for loop in iteration $\rho$, the last value written to $R_{1}$ is $\rho \mod 2$.

  In iteration $\rho$, process $2$ first write to its register before reading $R_{1}$. Therefore, as we assume that the $\rho$th write to $R_{1}$ precedes the $\rho$th write to $R_{2}$, it follows that process $2$ reads $\rho \mod 2$ from $R_{1}$. By the code, it proceeds to the next iteration, or terminates normally its for loop whenever $\rho = k$, a contradiction. 
\end{proof}

As a corollary,  there is only one case in which both processes perform the same number of iterations of the for loop: 
\begin{lemma}
  \label{lem:different_r}
  If $r_1 = r_2$ then $r_1 = r_2 =k$. 
\end{lemma}
\begin{proof}
  Suppose that $r_i < k$. By Lemma~\ref{lem:break}, $r_{3-i} \neq r_i$ since not both processes can break from  their for loop in the same iteration. 
\end{proof}

Next technical lemma establishes that when processes exits their for-loop at distinct iteration numbers, they both decide at line~\ref{line:dec_r<k}. 
\begin{lemma}
  \label{lem:k_k-1}
  If $\{r_1,r_2\} = \{k-1,k\}$, then no process decides at line~\ref{line:dec_r=k}. 
\end{lemma}
\begin{proof}
  Suppose without loss of generality that $r_1 = k$. By the code, this means that the last bit read in $R_{2}$ by process $1$ is $k \mod 2$. Since $r_2 = k-1$, the last value  written in $R_{2}$ is $k-1 \mod 2$. Hence, process $1$ must read $R_{2}$ in iteration $k$ before this last write by process $2$ occurs.

Since process $1$ reaches iteration $k$ of its for loop, 
 the sequence of values returned by the  reads of $R_{2}$ by process $1$ in iteration $1,\ldots,r_2-1$ is $(1,0,\ldots, (k-1) \mod 2)$. Therefore, as the initial value of $R_{1}$ is $0$, process $2$ must have written to $R_{1}$ at least $(k-1)$ times before process $1$ enters iteration $k$. As $r_2 = k-1$, process $2$ writes exactly $(k-1)$ times to $R_{2}$, and the last value it writes is $k-1 \mod 2$. It thus follows that in iteration $k$, process $1$ reads $k-1 \mod 2$ from $R_{2}$. Therefore, we have $new \neq k \mod 2$ and process $1$ cannot decide at line~\ref{line:dec_r=k}.   
\end{proof}

We are now ready to prove that the decisions of the two processes are at most $\frac{1}{2k+1}$ apart. The proof is a case analysis depending on the values $r_1$  and $r_2$. 
\begin{lemma}
  \label{lem:approx_agreement}
  $|y_1 - y_2| \leq \frac{1}{2k+1}$.
\end{lemma}

\begin{proof}
Let $x_i \in \{0,1\}$ and $y_i$ denote respectively the input and the decision of process $i$. In each case, we compute the value of $|y_1 - y_2|$ by following the code and we show that it is upper-bounded by $\frac{1}{2k+1}$. 
  We  start by  considering the case in which one process $i$ decides at line~\ref{line:decIdem}. Let $j = 3-i$ be the other process. By the code, this means  that process $i$ reads $\bot$ or $x_i$ in the other process' input register $I_{j}$.
  \begin{itemize}
  \item Process $i$ reads $\bot$ from $I_{j}$. It thus decides $y_i = x_i$. If $x_j = x_i$, process $j$ also decides at line~\ref{line:decIdem} and both processes decide the same value  $x_i$. Suppose that $x_j \neq x_i$. Since process $i$ reads $I_{j}$ after existing the for loop, and the first step of process $j$ is writing its input to $I_{j}$ (line~\ref{line:writeinput}), process $i$ exists the for loop in the first iteration, that is $r_i = 1$. By Lemma~\ref{lem:delta_r}, $r_j = 2$ and when process $j$ exits the for loop, $new = 1$. By the code, it thus follows that $j$ decides at line~\ref{line:dec_r<k} $y_j = (x_i + (\frac{(-1)^{x_i}}{2k+1})) = \frac{1}{2k+1}$ if $x_i = 0$ and $1 - \frac{1}{2k+1}$ otherwise.  Hence $|y_i - y_j| \leq \epsilon$.
  \item Process $i$ reads $x_i$ from $I_{j}$. It decides $y_i = x_i  =x_j$. Process $j$ either reads $\bot$ from $I_{i}$ and decides $y_j = x_i = x_j$. Otherwise, it sees both process $i$ input values and decides $x_i = x_j$ at line~\ref{line:decIdem} since both inputs are the same. 
  \end{itemize}
  
It remains to examine the cases in which no processes decide at line~\ref{line:decIdem}.   We first consider the case in which both processes decide after completing the same number of iterations, that is $r_1 = r_2$. By Lemma~\ref{lem:different_r}, $r_1 = r_2 = k$. A process reaching iteration $k$ may still exit the for loop prematurely (at line~\ref{line:break}) and decides at line~\ref{line:dec_r<k} or decide at line~\ref{line:dec_r=k} after completing the $k$ iterations. By Lemma~\ref{lem:break}, the two  processes cannot both decide at line~\ref{line:dec_r<k}. We examine all these cases next: 
    \begin{itemize}
    \item Both processes terminate their for loop normally. In that case they decide at line~\ref{line:dec_r=k}. Depending on the parity of $k$, we have two cases:
      \begin{itemize}
      \item $k \mod 2 = 0$: We have $y_1 = \frac{x_1 + k}{2k+1}$ and $y_2 = \frac{x_2 + k}{2k+1}$, and thus, since $x_1,x_2 \in \{0,1\}$, $$|y_1-y_2| = \frac{|x_1 - x_2|}{2k+1} \leq \frac{1}{2k+1}.$$
      \item $k \mod 2 = 1$: Similarly to the previous sub-case, we have $y_1 = \frac{x_2 + k}{2k+1}$ and  $y_2 = \frac{x_1 + k}{2k+1}$, from which it follows that, since $x_1,x_2 \in \{0,1\}$, $$|y_1-y_2| = \frac{|x_2 - x_1|}{2k+1} \leq \frac{1}{2k+1}.$$  
      \end{itemize}
    \item One process decides at line~\ref{line:dec_r=k} and the other at line~\ref{line:dec_r<k}. We have four cases: 
      \begin{itemize}
      \item $k \mod 2 = 0$ and process $1$ decides at line~\ref{line:dec_r=k}:
        $y_1 = \frac{x_1 + k}{2k+1}$ and $y_2 = \frac{x_1(2k+1) + (-1)^{x_1}(k-1)}{2k+1}$, from which we derive $$|y_1 - y_2| = \frac{x_1 +k -x_1(2k+1) - (-1)^{x_1}(k-1)}{2k+1} = \frac{1}{2k+1}.$$
      \item $k \mod 2 = 0$ and process $2$ decides at line~\ref{line:dec_r=k}:
        $y_2 = \frac{x_2 +k}{2k+1}$, and $y_1 = \frac{x_2(2k+1) + (-1)^{x_2}(k-1)}{2k+1}$, from which we derive $$|y_1 - y_2| = \frac{x_2(2k+1) + (-1)^{x_2}(k-1) - x_2 -k}{2k+1} = \frac{1}{2k+1}.$$
      \item $k \mod 2 = 1$ and process $1$ decides at line~\ref{line:dec_r=k}:
        $y_1 = \frac{x_2+k}{2k+1}$, and 
        $y_2 = \frac{x_2(2k+1) + (-1)^{x_2}(k-1)}{2k+1}$, from which we derive, as in the previous sub-case, $$|y_1 - y_2| = \frac{1}{2k+1}.$$ 
      \item $k \mod 2 = 1$ and process $2$ decides at line~\ref{line:dec_r=k}:
              $y_1 = \frac{x_1(2k+1) (-1)^{x_1}(k-1)}{2k-1}$, and $y_2 = \frac{x_1 +k}{2k+1}$, therefore,  as in the first sub-case, $$|y_1 - y_2| = \frac{1}{2k+1}.$$ 
      \end{itemize}
    \end{itemize}

    It remains to examine the cases in which $r_1 \neq r_2$. Note that by Lemma~\ref{lem:delta_r} $|r_1 - r_2| = 1$. Moreover, it follows from Lemma~\ref{lem:k_k-1} that both processes decide at line~\ref{line:dec_r<k}. Let us assume that $r_1 = r_2 + 1$ (the other case is similar). There are two sub-cases: 
    \begin{itemize}
    \item $r_1 \mod 2 = 0$: We have
      $y_1 = \frac{x_2(2k+1) + (-1)^{x_2}(r_1-1)}{2k+1}$, and 
      $y_2 = \frac{x_2(2k+1) + (-1)^{x_2}(r_1-2)}{2k+1}$, therefore 
      $$|y_1 - y_2| = \frac{|x_2(2k+1) + (-1)^{x_2}(r_1-1) - x_2(2k+1) - (-1)^{x_2}(r_1-2)|}{2k+1} = \frac{|(-1)^{x_2}|}{2k+1}.$$ 
    \item $r_1 \mod 2 = 1$. Similarly, we have
      $y_1 = \frac{x_1(2k+1) + (-1)^{x_1}(r_1-1)}{2k+1}$, and
      $y_2 = \frac{x_1(2k+1) + (-1)^{x_1}(r_1-2)}{2k+1}$, therefore
      $$|y_1 - y_2| = \frac{|(-1)^{x_1}|}{2k+1}.$$  
    \end{itemize}
    This case analysis completes the proof of the lemma. 
\end{proof}

Finally we remark that a process with input $x \in \{0,1\}$ never decides $1-x$ (see for example Figure~\ref{fig:1bit_eps_exe}: process $1$ with input $0$ (represented by black dots) cannot decide $1$.). This observation will be useful in the next section.

\begin{lemma}
  \label{lem:eps}
  For any $i \in \{1,2\}$, if process $i$ decides $y\in\{0,1\}$ then its input is $y$. 
\end{lemma}

We  now have all the ingredients to prove the correctness  of the algorithm. 

\begin{proposition}\label{prop:Ak}
  For every $k\geq 0$, Algorithm ${\cal A}_{k}$ solves $\frac{1}{2k+1}$-approximate agreement  wait-free  for two processes using 1-bit registers with worst-case step complexity $O(k)$.
\end{proposition}
\begin{proof}
  By the code,  each process performs at most $O(k)$ write and read operations, and thus the algorithm is wait-free. The decisions are valid. If both processes have the same input, each process that decides decides its own input at line~\ref{line:decinit}. Otherwise, if the set of input is $\{0,1\}$, every decision (obtained at line~\ref{line:decinit}, line~\ref{line:dec_r<k} or line~\ref{line:dec_r=k}) belongs to $\{0,\frac{1}{2k+1},\ldots,\frac{2k}{2k+1},1\}$. Finally, it follows from Lemma~\ref{lem:approx_agreement} when the two processes decide, their output values are at most $\frac{1}{2k+1}$ apart from each other. 
\end{proof}

\subsection{From $\epsilon$-Agreement to Arbitrary Tasks}

In the previous section, we have seen that with two $1$-bit registers, two processes can wait-free solve $\epsilon$-agreement where $\epsilon$ is arbitrary small. 
Note that for two processes, an algorithm that solves a task has to be $1$-resilient, that is  has to tolerate at most one crash. Reciprocally, any $1$-resilient two processes algorithm is also wait-free. It is known that $\epsilon$-agreement plays a  central role in the characterization of $1$-resilient solvable tasks. In fact, in \cite{BMZ90charac}  an \emph{universal} protocol is presented, for solving any $1$-resilient solvable task which relies at its heart on an  $\epsilon$-agreement protocol. We show that this approach  also applies to the case of two processes communicating with constant  size  registers.

Biran, Moran and Zaks~\cite{BMZ90charac} give necessary
and sufficient conditions for tasks to be solvable
$1$-resiliently. Essentially,  solving a task   $\Pi = (\m{I},\m{O},\Delta)$ boils down for the processes to select an edge in a graph induced by the
possible outputs of the task. Each node represents a valid output (a couple of values $(y_1,y_2)$ in the case of two processes) and two nodes are neighbours if the corresponding outputs differ by one value. Each process is required to select a node such that the set of selected nodes is a singleton $\{v\}$ or form an edge $\{v,v'\}$ in the graph. Each process $i$ then decides the  value $y_i$  in the output associated with its selected node. The couple $(y_1,y_2)$ decided by the processes is a valid output for $\Pi$ since it is the output associated with $v$ or $v'$. We explain next how this can be achieved with $\epsilon$-agreement, and how we make sure that the decision $(y_1,y_2)$ is legal with respect to the specification of the task $\Delta$ and the inputs of the processes.

\subsubsection{Graph-theoretic characterization of $1$-resiliency}

Biran et al.~\cite{BMZ90charac} characterized the tasks that are $1$-resilient solvable in the asynchronous message passing model. It is not hard to see that the characterization also holds for the shared memory model in which processes communicate with read/write registers. For the  necessary conditions, the only property of the message passing model used in the proof is $i$-sleeping execution in which  all messages by  process $i$ are  arbitrarily delayed. This translates in shared memory by arbitrarily delaying  a write operation by process $i$. The sufficient condition is based on  the universal protocol mentioned above, which can be simulated in shared memory.

We recall the main ingredients of the characterization of \cite{BMZ90charac}, adapted to the case of two processes. 
Let $\Pi = (\m{I},\m{O},\Delta)$ be a task for
two processes.  An input $X \in \m{I}$ (respectively, an output $Y \in \m{O}$) for $\Pi$ is a couple of values $X = (x_1,x_2)$ (respectively, $Y = (y_1,y_2)$). As a process may fail before deciding or writing  its input to a shared register, \cite{BMZ90charac}  considers \emph{partial} input $X^i$ or output $Y^i$ in which the input or output of process $i, i \in \{1,2\}$ is missing and replaced by $\bot$. $X$ (respectively, $Y$) is then an \emph{extension} of $X^i$ (respectively, of $Y^i$).

Given a set $\m{O}' \subseteq \m{O}$ of outputs,  $G(\m{O}')$ is the graph whose
set of nodes is $\m{O}'$ and there is an edge between two nodes $y$ and
$y'$ if and only if $y$ and $y'$ differ in only one value. Said differently, the output value of one process $i\in\{1,2\}$ is the same in $y$ and $y'$. Essentially, $\Pi$ is $1$-resilient solvable if (a subset of) its outputs is \emph{connected} and, for any partial input $X^i$, there exits a partial output $Y^i$ that is \emph{extendable}: For any extension $X$ of $X^i$, there exists an extension $Y \in \Delta(X)$ of $Y^i$. 

\begin{lemma}[\cite{BMZ90charac}]
  \label{lem:1res_charac}
 Task $\Pi = (\m{I},\m{O},\Delta)$ is $1$-resilient solvable if and only if there exists 
 a subset $\m{O'} \subseteq \m{O}$ of the outputs satisfying:
 \begin{description}
 \item[Connectivity:] For every  input $X \in \m{I}$, the graph
   $G(\Delta(X) \cap \m{O'})$ is connected, and
\item[Covering:] For every partial input  $X^i$, $i\in \{1,2\}$, there
  exists a partial output  $Y^i$ such that, for every extension $X \in \m{I}$ of~$X^i$, 
  there exists an extension $Y$ of $Y^i$  in $\Delta(X) \cap \m{O}'$. 
\end{description}  
\end{lemma}

\subsubsection{Solving any $1$-resilient solvable task}

Based on Lemma~\ref{lem:1res_charac}, a protocol for arbitrary $1$-resilient solvable task is derived in~\cite{BMZ90charac} as follows. Given an input $X \in \m{I}$ and a partial input $X^i$, for some $i \in \{1,2\}$, one can find a path ($Y_0,\ldots,Y_{L}$)
in $G(\m{O}')$ such that:
\begin{itemize}
\item $Y_0, \ldots, Y_{L-1}$ are valid outputs for $X$, i.e., $Y_i \in \Delta(X)$ for any $i \in \{0,\ldots, L-1\}$ 
\item $Y_{L-1}$ and $Y_{L}$ differ in their $i$th  component (that is, $Y_{L-1}$ is an extension of $Y_{L}^i$). 
\end{itemize}
$Y_0$ is an  output arbitrarily chosen  in $\Delta(X) \cap \m{O}'$. In particular, the choice of $Y_0$ does not depend on $i$. We next pick outputs $Y_{L}$ and $Y_{L-1}$. 
By the covering condition, there exists a partial output $Y^i$ that can be extended to a valid output for any extension of $X^i$. Let $Y_{L} \in \m{O}'$ be any extension of $Y^i$. Note that $Y_{L}$ may not be a valid output for $X$, that is it is not required that $Y_{L} \in \Delta(X)$. Let also $Y_{L-1}  \in \m{O}'$ that extends $Y^i$ and is valid for input $X$, i.e., $Y_{L-1} \in \Delta(X) \cap \m{O}'$.

We now explain why there is a path between $Y_0$ and $Y_{L}$ in $G(\m{O}')$. Since both $Y_{L-1}$ and $Y_{L}$ extend the same partial output $Y^i$, they are neighbours in $G(\m{O}')$.  Moreover, as $Y_0$ and $Y_{L-1}$ both belong to $\Delta(X) \cap \m{O}
'$, there are connected in $G(\m{O}')$ by the connectivity condition of Lemma~\ref{lem:1res_charac}.

Let $\m{I}^i$ be the set of partial inputs in which the input of process $i$ is missing, and 
let $$\delta : \m{I} \cup \m{I}^1 \cup \m{I}^2 \to \m{O'}$$ a map that associates with each (partial) input an output as above. We for any pair of input/partial input $X, X^i$, we fix a path in $G(\m{O}')$, denoted  $$path(\delta(X),\delta(X^i)) = (Y_0,\ldots,Y_{L}),$$  connecting $\delta(X)$ to $\delta(X^i)$. For simplicity, we assume that all paths have the same length $L\geq 2$. This can be achieved by adding adjacent copies of one of the nodes in shorter paths.

A key property of $path(\delta(X),\delta(X^i)) =(Y_0,\ldots,Y_{L})$ is that when the input is $X$, processes can safely pick their output in any pair of adjacent nodes.  Suppose for example that $Y = (y_1,y_2)$ and  $Y' = (y_1',y_2')$  are adjacent nodes in the path, process $1$ picks its output in $Y$ and process~$2$ in $Y'$.  There is only one output value that differs between $Y$ and $Y'$, say $y_2 \neq y_2'$. Then the global output is $Y' \in \Delta(X)$ by definition of the path. The only exception is when $Y$ or $Y'$ is equal to~$Y_{L}$, which may not be in $\Delta(X)$. But in that case we make sure that process $j \neq i$ does not pick its output in $Y_{L}$.

Having processes selecting adjacent nodes is easily done with an instance $\epsilon$-agreement, with $\epsilon = \frac{1}{L}$. Each process starts the $\epsilon$- agreement protocol with $0$ (if it sees the full input) or $1$ (if it sees only its input) and gets back a decision $d = k \epsilon$ for some $k \in \{0,\ldots, L\}$. It then selects node $Y_{d/\epsilon}$. When the memory is unbounded, as in~\cite{BMZ90charac}, processes can simply write in shared memory their view of the input (full or partial), from which the path $path(\delta(X),\delta(X^i))$ can be inferred. A little more care is required when the shared memory space is bounded.

\subsubsection{Solving any $2$-processes wait-free solvable task  with constant size registers}

Algorithm~\ref{alg:universal} presents an universal construction for solving any tasks solvable by two processes wait-free with registers of constant size. It relies on the $\epsilon$-agreement protocol (Algorithm~\ref{alg:epsilon}) of Section~\ref{sec:epsilon_agreement}, and requires 2 registers of size $3$ bits, one per process. 

The initial value of each  input register $I_{1}$ and $I_{2}$ is $\bot$. Algorithm~\ref{alg:universal} does not explicitly use other registers. However, processes invoke the $\epsilon$-agreement protocol (Algorithm~\ref{alg:epsilon}) in which two registers per process are needed. In Algorithm~\ref{alg:epsilon}, process $i$ requires an input register, whose value is either $\bot,0$ or $1$ and a binary communication register (register $R_{i}$). Since  only process $i$ writes to these registers, they could be replaced by a single register of size $3$ bits. 

Given a couple $C = (a,b)$, let $C[1]$ denotes the first value $a$ and $C[2]$, the second value $b$. Processes are equipped with the same map $\delta$ from the inputs and partial inputs of $\Pi$ to $\m{O}$, as well as the same collection of paths $$\{\mathit{path}(\delta(X),\delta(X^i)) : X \in \m{I}, i \in \{1,2\}\}.$$ As explained previously, we may assume that all paths have the same length $L$.
Process $i$ first writes its input $x_i$ to the special register $I_{i}$ (line~\ref{lu:writeinput}) and reads the input of the other register. Its view (partial if it reads $\bot$ in $I_{other}$ or full otherwise) of the global input to task $\Pi$ is encoded in its input to the $\epsilon$-agreement protocol (line~\ref{lu:epsinput}). 
Process $i$ then invokes the $\epsilon$-agreement protocol of Section~\ref{sec:epsilon_agreement} and gets back a value $d \in \{0,\frac{1}{L},\ldots,\frac{L-1}{L},1\}$ (line~\ref{lu:epsinvoke}).
Depending on $d$, the process then computes a partial input (line~\ref{lu:partialX1}, line~\ref{lu:partialX2} or line~\ref{lu:case1}) and/or a full input (line~\ref{lu:case0} and line~\ref{lu:fullX}). The full input $fullX$, when computed by both processes, is the same and equal to $(x_1,x_2)$. The same holds for $partialX$. Indeed, when process $i$ sets $partialX$, it is certain that the other process has a different view than its own of the task input before invoking $\epsilon$-agreement.

For example, suppose that process $1$ gets back $d, 0\leq d < 1$ from the $\epsilon$-agreement protocol to which its input $myInput$ is $0$. Then it knows that the input of process $2$  to $\epsilon$-agreement is $1$ (otherwise, both processes  decide $0$). Hence process $2$ misses $x_1$ when it first reads the input register. Consequently, $partialX$ is set to $(\bot,x_2)$ at process $1$ (line~\ref{lu:partialX2}). For process $2$, $partialX$ is also $(\bot,x_2)$ if it is set at line~\ref{lu:partialX1} (since for process $2$, $myInput=1$) or at line~\ref{lu:case1} (since in that case, it does not know the input of process $1$).
Therefore, both processes are able to select the same path $$path(\delta(fullX),\delta(partialX)) = Y_0,\ldots,Y_{L}.$$ Process $i$ then returns $y_i$, the value at position $i$ in $Y_{dL}$ (line~\ref{lu:decidefull}, line~\ref{lu:decideY} or line~\ref{lu:decidepartial}).    As $\epsilon = \frac{1}{L}$, the output $Y,Y'$ chosen by the processes are consecutive in path. They thus differ in at most value, and hence $(y_1,y_2) = Y \in \{Y_0,\ldots,Y_{L}\}$. We show in the proof that $Y \neq Y_{L}$ and therefore, $Y \in \Delta(X)$ by construction of the path $path(\delta(fullX),\delta(partialX))$. 

\begin{algorithm}[tb]
  \caption{Solving an arbitrary wait-free  solvable task $\Pi =(\m{I}, \m{O}, \Delta)$ with $3$-bit registers for two processes. Code for process $i \in \{1,2\}$ with input $x_i$}
  \label{alg:universal}
  \begin{algorithmic}[1]
    \Statex \textbf{pre-processing:} \Comment{Common to both processes}
    \Statex \hspace{\algorithmicindent} compute map $\delta:\m{I}\cup\m{I}^1\cup\m{I}^2 \to \m{O}$, integer $L$, 
    \Statex \hspace{\algorithmicindent} and $\mathit{path}(\delta(X),\delta(X^k))$ for every $X\in \m{I}$, and every $k \in \{1,2\}$
    \Statex \textbf{local variables:}
    \Statex \hspace{\algorithmicindent} $me \gets i$  \Comment{my index}
    \Statex \hspace{\algorithmicindent} $other \gets 3-i$ \Comment{index of the other process}
    \Statex \hspace{\algorithmicindent} $myInput \in \{0,1\}$         \Comment{input of process $i$ for the $\epsilon$-agreement protocol}

    \Begin
    \State $I_{me}.\mathsf{write}(x_{me})$ \label{lu:readinput} \Comment{write input for task $\Pi$}
    \State $x_{other}  \gets I_{other}.\mathsf{read}()$  \label{lu:writeinput}\Comment{Read other proc. input}
    \If{$x_{other} = \bot$} $myInput \gets 1$ \textbf{else} $myInput \gets 0$\label{lu:epsinput}    \EndIf
    \State $d \gets \m{A}_\frac{L}{2}(myInput)$\label{lu:epsinvoke}       \Comment{$\epsilon$-agreement (Algorithm~\ref{alg:epsilon}) with  $\epsilon = \frac{1}{L+1}$}
    \If{$d = 0$}
    		\State $(fullX[me],fullX[other]) \gets (x_{me},x_{other})$ \label{lu:case0} 
    		\State \textsl{return} $\delta(fullX)[me]$ \label{lu:decidefull}  \Comment{return value at position $me$ in $\delta(fullX)$}
    \Else
    		\If{$0 < d < 1$}
			\State $x_{other} \label{lu:cased}  \gets I_{other}.\mathsf{read}()$ 
   			\State  $(fullX[me],fullX[other]) \gets (x_{me},x_{other})$  \label{lu:fullX}
			\If{$myInput = 1$} 
				\State $partialX \gets \mathit{fullX}^{other}$ \label{lu:partialX1}
			\Else  \Comment{replace value at position $other$ or $me$ with $\bot$ in $fullX$}
				\State $partialX \gets \mathit{fullX}^{me}$ \label{lu:partialX2}
			\EndIf
    			\State $(Y_0, \ldots, Y_{dL}, \ldots, Y_{L}) \gets path(\delta(fullX),\delta(partialX))$
    			\State \textsl{return} $Y_{dL}[me]$ \label{lu:decideY} \Comment{return value at position $me$ in $Y_{dL}$}
   		 \Else\Comment{final case: $d=1$}
    			\State  $(partialX[me],partialX[other]) \gets (x_{me}, \bot)$ \label{lu:case1}
   			 \State  \textsl{return}($\delta(partialX)[me]$) \label{lu:decidepartial}   
			 				\Comment{return value at position $me$ in $\delta(partialX)$}
		  \EndIf
   	\EndIf
\End
\end{algorithmic}
\end{algorithm}

\subsubsection{Correctness of Algorithm~\ref{alg:universal} and proof of Theorem~\ref{theo:case-of-two-proc}}

Let $\Pi = (\m{I},\m{O},\Delta)$ be a task solvable by two processes wait-free with unbounded registers. As $\Pi$ is $1$-resilient solvable by processes $1$ and $2$, we have seen that there exists a map $$\delta : \m{I} \cup \m{I}^1 \cup \m{I}^2 \to \m{O},$$ an integer $L \geq 3$ and a collection $\{path(\delta(X),\delta(X^i)): X \in \m{I}, i\in \{1,2\}\}$ of sequences of $L+1$ outputs in $\m{O}$  such that, for every 
$X \in \m{I}$, and every  $i \in \{1,2\}$, if  $$path(\delta(X),\delta(X^i)) = Y_0,\ldots,Y_{L}$$ then, for every $j\in\{1,\dots,L\}$
\begin{itemize}
\item $Y_j \in \Delta(X)$,
\item $Y_j$ and $Y_{j+1}$ differ in at most one value,
\item $Y_{L-1}$ and $Y_{L}$   differ in their value at position $i$.  
\end{itemize}

\begin{lemma}
  \label{lem:ucorrect}
  Let $y_1,y_2$ be the value returned by processes $1$ and $2$   in an execution of Algorithm~\ref{alg:universal} in which their input is $x_1,x_2$, respectively.  $Y = (y_1,y_2) \in \Delta(X)$, where $X = (x_1,x_2)$. 
\end{lemma}

\begin{proof}
  Each process $i, i \in \{1,2\}$ returns $y_i$ at line~\ref{lu:decidefull}, line~\ref{lu:decideY} or line~\ref{lu:decidepartial} depending of the value it obtains from the $\epsilon$-agreement protocol. Let $d_i$ be the value obtained from the $\epsilon$-agreement protocol by process $i$ (line~\ref{lu:epsinvoke}). 
  \begin{itemize}
  \item $d_1 = d_2 = 0$. By Lemma~\ref{lem:eps}, the input of both process to the $\epsilon$-agreement protocol is $0$. They thus both see the full input $X = (x_1,x_2)$ of task $\Pi$, that is they read a non-$\bot$ value in the other process register (line~\ref{lu:readinput}). Hence, $Y = (y_1,y_2) = \delta(X)$, which belongs to $\Delta(X)$ by definition of $\delta$. 
  \item $d_1= d_2 = 1$. This case is not possible. By Lemma~\ref{lem:eps}, it would imply that the input of both processes to the $\epsilon$-agreement protocol is $1$. But the last process that reads the other process' input register (line~\ref{lu:readinput}) sees the other process input and thus sets its input to the $\epsilon$-agreement protocol to $0$ (line~\ref{lu:epsinput}).
  \item It remains to examine the case in which $ 0 < d_i < 1$ for some process $i$. Without loss of generality, let us assume that $0 < d_1  < 1$. Processes have therefore invoke the $\epsilon$-agreement protocol with distinct inputs before process $1$  gets back $d_1$. Hence, when process $1$ reads the input register of process $2$ at line~\ref{lu:cased}, it obtains $x_2 \neq \bot$. Thus $\mathit{fullX}$ is a full input in $\m{I}$, equal to $(x_1,x_2)$. $partialX$ is either $(x_1,\bot)$ or $(\bot,x_2)$   depending on the input of process $1$  to the $\epsilon$-agreement protocol. If process $1$ has not read $x_2$ at line~\ref{lu:readinput}, $partialX = (x_1,\bot)$ and $partialX = (\bot,x_2)$ otherwise.

Observe  that if process $2$ also computes a partial output (at line~\ref{lu:partialX1}, at line~\ref{lu:partialX2} or line~\ref{lu:case1}), it computes the same partial output. For example, if process $1$ has not read $x_2$ at line~\ref{lu:readinput}, it must be the case that process $2$ reads $x_1$ when it performs line~\ref{lu:readinput}. The partial output computed by process $2$ is then also $(x_1,\bot)$. Similarly,  if process $1$ has  read $x_2$ line~\ref{lu:readinput}, as both processes have distinct input to the $\epsilon$-agreement protocol, process $2$ misses the input of process $1$ at line~\ref{lu:readinput}. Therefore,  the partial output computed by process $2$ is $(\bot,x_2)$. 

Let $$d_1 = k_1\epsilon = k_1/L,$$ and let us  consider $$path(\delta(fullX),\delta(partialX)) = (Y_0, \ldots,Y_{L}).$$ 

Process 2 obtains $d_2 = k_2/L$ from the $\epsilon$-agreement protocol with $k_2 \in \{k_1-1,k_1,k_1+1\}$. The value $y_1,y_2$ returned by the processes are thus taken from the same outputs or two consecutive outputs $Y_k,Y_{k+1}$. Since consecutive outputs in $path(\delta(fullX,partialX))$ differ by at most one value, $(y_1,y_2) =Y_k$ or $(y_1,y_2) =Y_{k+1}$.  Finally, $$(y_1,y_2) \neq Y_{L} = \delta(partialX)$$ as this would mean that both processes obtain $1$ from the $\epsilon$-agreement protocol. Therefore, there exists $j \in \{1,\ldots,L\}$ such that $(y_1,y_2) = Y_j$. Hence, $$(y_1,y_2) = Y_j \in \Delta(fullX) = \Delta(x_1,x_2).$$ by construction of  $path(\delta(fullX,partialX))$.  \qedhere
 \end{itemize}
\end{proof}

We are now ready to prove Theorem~\ref{theo:case-of-two-proc}.
\begin{proof}[Proof of Theorem~\ref{theo:case-of-two-proc}]
  Algorithm~\ref{alg:universal} is wait-free, as the $\epsilon$-agreement protocol (Algorithm~\ref{alg:epsilon}) is wait-free and it has no loop. Besides input register, it requires a 3-bits register per process $i$ (for emulating the input register, which contain $\bot,0$ or $1$ and the binary register $R_i$ of Algorithm~\ref{alg:epsilon}). 
Let $\Pi = (\m{I},\m{O},\Delta)$ is a task solvable by two processes wait-free with unbounded registers.  By Lemma~\ref{lem:ucorrect},   Algorithm~\ref{alg:universal} correctly solves $\Pi$.
\end{proof}
\section{Universality when a Minority of Processes May Fail}
\label{sec:minority}

Theorem~\ref{theo:small-values-of-t} is a direct consequence of the following result. 

\begin{proposition}
There exists $c$ such that, for any $n\geq 2$, and any $t$ with $1 \leq t < \frac{n}{2}$, every task solvable by $n$ processes in the $t$-resilient shared-memory model with unbounded-size registers can be solved  by $n$ processes  in the $t$-resilient shared-memory model with registers of  $c(t+1)$ bits. 
\end{proposition}

\begin{proof}
Let $\Pi$ be a task solvable by $n$ processes in the $t$-resilient shared-memory model with unbounded-size registers, and let $\m{A}$ be an algorithm solving $\Pi$. The construction of an algorithm $\m{B}$ solving $\Pi$ with registers of size $c(t+1)$ bits proceeds in three phases, essentially combining known results from the literature. 

First, since $t<\frac{n}{2}$, it is known that there is an $n$-processes \emph{message-passing} algorithm $\m{A}'$ solving $\Pi$ when at most $t$ processes may crash. Specifically, $\m{A}'$ can be obtained by the message-passing simulation of shared-memory systems (see~\cite{AttiyaBD95}).
In the latter, in addition to the value being written or read, each message carries some ``control bits''. 

Second, the $t$-resilient message-passing model can be simulated by a network with fewer communication links between processes,  as long as the network is $(t+1)$-connected, i.e., removing at most $t$ nodes of the network keeps it connected. In particular, the message-passing model can be simulated in the network displayed in Figure~\ref{fig:t-ring}, which we call \emph{$t$-augmented ring}. In the $t$-augmented ring, nodes $1,\dots,n$ form a directed cycle (displayed in blue on Figure~\ref{fig:t-ring}). Every node $i\in [n]$ has $t$-additional  out-neighbors, nodes $(i+1) \bmod n+1, \ldots, (i+t) \bmod n+1$, and thus $t$ additional in-neighbors as well, nodes  $(i-3) \bmod n+1, \ldots, (i-1 - (t+1)) \bmod n+1$ (the corresponding links are displayed as dashed arrows on Figure~\ref{fig:t-ring}).

For every $i\in [n]$, when sending (or receiving) a message~$m$ to (from) some process~$j\neq i$,  Process~$i$ sends  $m$ directly to $j$ if there is a link $(i,j)$, otherwise it sends it to all its $t+1$ successors in the $t$-augmented ring. From these latter nodes, the message $m$ is forwarded using the same mechanism until it reaches its destination. As at most $t$ processes may fail, any message sent by a non-faulty process is eventually received by its destination whenever the latter is not faulty. 

\begin{figure}[tb]\centering
\begin{tikzpicture}
\usetikzlibrary {graphs.standard}  
 \graph [->,nodes={draw, thick,  circle}, edges={thick, blue}, grow right=2cm] {
  1 -> 2 -> 3 -> 4 -> 5 -> 6 -> 7;
  7 -> [edge={bend right}] 1;
  1 ->[edge={bend right, dashed, gray}] {3,4};
  2 ->[edge={bend left, dashed, gray}]  {4,5};
  3 ->[edge={bend right, dashed, gray}] {5,6};
  4 ->[edge={bend right, dashed, gray}]  {6,7};
  5 ->[edge={bend right, dashed, gray}]  {7,1};
  6 ->[edge={bend right, dashed, gray}]  {1,2};
  7 ->[edge={bend right, dashed, gray}]  {2,3};
};
\end{tikzpicture}
\caption{\sl The $2$-augmented $7$-node ring}
\label{fig:t-ring}
\end{figure}
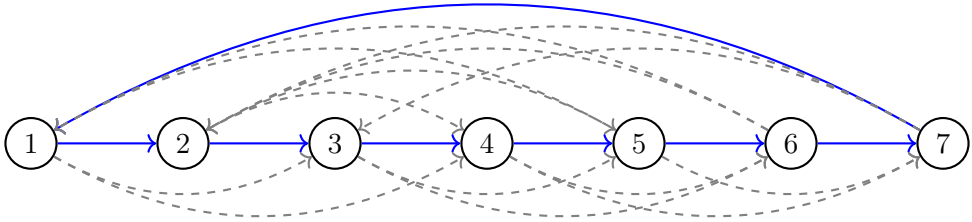

Finally,  we rely on the  \emph{alternating bit} protocol (see~\cite{bartlett1969note,lynch1968computer}), to simulate the $n$-node $t$-augmented ring by a shared memory, with registers of size $3(t+1)$ bits.
For that, we implement message passing communication on channel from $i$ to $j$ by a 2~bits register $R_i$ written by $p_i$ and read by $p_j$ and a 1~bit register $A_j$ (initialized to 0) written by $p_j$ and read by $p_i$.
First of all, message $m$ to be sent from $p_i$ to $p_j$ is encoded by inserting  $0$ between each bit and adding a $1$ at the end, giving a message $m'$. In this way, $p_j$ will be enable to determine the end of the message and get $m$. 

Following the alternating bit protocol, $p_i$ maintains a local bit $a_i$ (initialized to 0).  $m'$ will be transmit bit by bit, one after the other. To send a bit $b$ of $m'$, $p_i$ reads register $A_j$ until it reads value $a_i$ then it writes $(b,a_i)$ in $R_i$ and flip $a_i$ (hence it sends the next bit only when register $A_j$ has flipped). To receive a bit $p_j$ reads $R_i$ until it reads $(b,a)$ with $a=A_j$ for some $b$ then $p_j$ flips $A_j$ (acknowledging the received bit).  Hence the sequence of bits read by $p_j$ is $m'$ that can easily be decoded in $m$.

For each of its $t+1$ output links  $p_i$ writes in a 2~bits register and for each of its input links it writes a 1~bit register. Regrouping this in a single register we get a $3(t+1)$ bits register.
\end{proof}


\section{Universality of the Iterated Immediate Snapshot (IIS) Model}
\label{sec:iterated}

This section is entirely dedicated to the proof of Theorem~\ref{theo:iterated-immediate-snapshot}. It directly follows from Propositions~\ref{prop:IC2IIS-1-bit} and~\ref{prop:simIS_in_IC} given below. 

\paragraph{Preliminaries.} 

Recall that, in the IC model, the instruction $M_r.\mathsf{collect}()$ returns an $n$-dimensional vector $\mathbf{v}$ in which the value $\mathbf{v}[i]$ of each entry $i\in [n]$ is either $\bot$, or a value from some domain $D$, where $\bot \notin D$. We define a partial order $\subset$ on such vectors as follows. For two vectors $\mathbf{v}$ and $\mathbf{v}'$, we set $\mathbf{v}\subset \mathbf{v}'$ if and only if the following three conditions hold:
\begin{itemize}
\item  For every $i\in \{1,\dots,n\}$ : $\mathbf{v}'[i] = \bot \implies \mathbf{v}[i] = \mathbf{v}'[i]$ 
\item For every $i\in \{1,\dots,n\}$ :  $\mathbf{v}[i] \neq \bot \implies \mathbf{v}'[i] = \mathbf{v}[i]$  
\item There exists $i\in \{1,\dots,n\}$ such that $\mathbf{v}[i] = \bot$ and $\mathbf{v}'[i] \neq \bot$.
\end{itemize}
For every $r\geq 1$, the vectors  $\mathbf{v}_1,\ldots,\mathbf{v}_n$ returned by $\mathsf{snapshot}()$ or $\mathbf{collect}()$ operations performed in memory $M_r$ by all processes $1,\ldots,n$ have a specific structure. Specifically, for every $i\in \{1,\dots,n\}$, let $x_i$ be the value written in memory~$M_r$ by~$p_i$. In case of $\mathsf{snapshot}()$, we have 
\begin{itemize}
\item Validity: For every $i,j\in \{1,\dots,n\}$, $\mathbf{v}_i[j] \in \{x_j, \bot\}$
\item Self-containment: For every $i\in \{1,\dots,n\}$, $\mathbf{v}_i[i] \neq \bot$
\item Inclusion: For every $i,j\in \{1,\dots,n\}$, $\mathbf{v}_i \subseteq \mathbf{v}_j$ or $\mathbf{v}_j \subseteq \mathbf{v}_i$
\end{itemize}
In case of $\mathsf{collect}()$, the validity and self-containment properties holds, but not necessarily the Inclusion property, which is replace by the weaker 
\begin{itemize}
\item Write-order consistency: For every $i,j\in \{1,\dots,n\}$, if the execution of $M_r[i].\mathsf{write}(x_i)$ of~$p_i$ precedes the execution of $M_r[j].\mathsf{write}(x_j)$ of~$p_j$ then $\mathbf{v}_j[i] = x_i$. 
\end{itemize}

\subsection{From IC with Unbounded Registers to IS with $1$-Bit Registers}
\label{sec:1bitIS_to_IC}

Without loss of generality, any algorithm \m{A} solving some task $\Pi$ in the iterated model, whether it be using collect, or snapshot,   can be assumed to be a full information algorithm following the generic pattern of  algorithm~\ref{alg:genericIC}. This algorithm is for the iterated collect model, and the one for the iterated snapshot model is obtained by simply replacing Instruction~\ref{instruc:snap2collect} by $$myview \gets M_r.\mathsf{snapshot}().$$ Only the number $k$ of iterations, and the function $\mathsf{decide}$ depend on the task~$\Pi$ to be solved. The following result is our first step for establishing Theorem~\ref{theo:iterated-immediate-snapshot}:

\begin{algorithm}[tb]
  \caption{Full-information protocol in the iterated collect (IC) model: Code of process~${i\in [n]}$} 
  \label{alg:genericIS}
  \label{alg:genericIC}
  \begin{algorithmic}[1]
  \Statex\textbf{shared memory:} 
  \Statex \hspace{\algorithmicindent} for every $r\geq 1$, $M_r$ is an array of $n$ SWMR registers, each initialized to~$\bot$
   \Begin 
   \State $myview[i] \gets $ input of process $i$
   \For{every $j\neq i$} $myview[j] \gets\bot$  \EndFor
    \For{$r=1,\ldots,k$} \Comment{$k$ is a fixed integer, which depends on the task to be solved}
    \State $M_r[i].\mathsf{write}(myview)$
    \State $myview \gets M_r.\mathsf{collect}()$ \label{instruc:snap2collect}
    \EndFor
    \State \textsl{return}($myview$)
    \End
  \end{algorithmic}
\end{algorithm}

\begin{proposition}\label{prop:IC2IIS-1-bit}
  For every task~$\Pi$, if $\Pi$ solvable in the IC model then $\Pi$ is solvable in IIS model with 1-bit registers.
\end{proposition}

The rest of the sub-section is dedicated to the proof of this proposition. Let \m{A}  be an algorithm solving $\Pi$ in the IC model, written in the generic form of algorithm~\ref{alg:genericIC}. We show that \m{A} can be implemented in the  IIS model with 1-bit registers. Given an execution of \m{A},  let us denote by $V_i^r$ the ``view'' of process~$i$ at the end of round~$r$, i.e., the content of its local variable $myview$ at the end of round~$r$. For $r\geq 1$, we define the \emph{round-$r$ configuration} 
$
(V_1^r,\ldots,V_n^r).
$
An \emph{initial configuration}  $ (V_1^0,\ldots,V_n^0)$ is an $n$-uplet formed by input values to the processes. Let $\m{C}^0$ be the set of initial configurations, and let $\m{C}^r$ be the set of  configurations after $r$ rounds, taken over all possible executions of $\m{A}$, and all possible inputs assignments. That is, for each $c \in \m{C}^r$, there is an input configuration~$I$, and an execution $\alpha$ of \m{A} starting from $I$ in which, for each process~$i$, the value of $V_i$ at the end of round $r$ is $c[i]$. Note that, as $\Pi$ has finitely many inputs, each set $\m{C}^r $ for $r \geq 0$ is finite. Let 
\[
\m{C} = \bigcup_{0 \leq r \leq k} \m{C}^r
\]
where $k$ is the number of write-collect iteration of~$\m{A}$ in its generic form. We order the configurations in $\m{C}$ in a \emph{round-preserving} manner. That is, for every $r\geq 0$, any configuration in $\m{C}^r$ is ordered before any configuration in $\m{C}^{r+1}$. That is,  we enumerate the configurations in $\m{C}$ as 
\begin{equation}\label{eq:lesconfigordonnees}
\m{C} = (c_1,\ldots,c_N)
\end{equation}
where $N=|\m{C}|$, and, for every $i,j\in \{1,\dots,N\}$, if  $c_i \in \m{C}^r$ and $c_j \in \m{C}^{r'}$ with  $r<r'$ then $i < j$. 

\begin{algorithm}[tb]
  \caption{Simulation of Algorithm~\ref{alg:genericIS} in the IS model with 1-bit registers. Code for process~${i\in [n]}$} 
  \label{alg:simulateIC}
  \begin{algorithmic}[1]
   \Statex\label{itsim:var}\textbf{shared memory:} 
   \Statex \hspace{\algorithmicindent} for every $r\geq 1$, $M_r$ is an array of $n$ SWMR 1-bit registers, each initialized to~$0$
  
    \Begin\label{itsim:begin}
    \State $W_i^0 \gets$ input of process $i$ \label{itsim:input}
    \State $\rho \gets 0$ \Comment{$\rho$ is an integer, indexing the iterations}
    \For{$r=1,\ldots,k$}\label{itsim:extloop}  \Comment{$k$ is the same as in Algorithm~\ref{alg:genericIS} }
            	\State $W_i^r \gets \varnothing$ \label{itsim:Wempty} 
     		\For{$\rho = \sum_{0 \leq \ell < r-1} |\m{C}^\ell|+1,\ldots,\sum_{0 \leq \ell < r} |\m{C}^\ell|$}\label{itsim:intloop} 
            								\Comment{$|\m{C}^{r-1}|$ iterations of $1$-bit snapshots}
       			\If{$c_\rho[i] = W_i^{r-1}$}\Comment{Configurations are ordered as in Eq.~\eqref{eq:lesconfigordonnees}}
       				\State $M_\rho[i].\mathsf{write}(1)$  \label{itsim:write}  
       			\Else  
       				\State $M_\rho[i].\mathsf{write}(0)$ 
      			 \EndIf
       			\State $s \gets M_\rho.\mathsf{snapshot}()$
       			\State  $W_i^r \gets W_i^r \cup \{c_\rho[j] : s[j] = 1\}$ \label{itsim:snap} 
       		\EndFor
    \EndFor
    \State \textsl{return}($W_i^N$) \label{itsim:decide} 
    		\Comment{$N=\sum_{0 \leq \ell \leq k} |\m{C}^\ell|=$ total number of configurations in Eq.~\eqref{eq:lesconfigordonnees}}
    \End\label{itsim:end}
  \end{algorithmic}
\end{algorithm}

\paragraph{The simulation algorithm.} 

Algorithm~\ref{alg:simulateIC} simulates an execution of $\m{A}$ in the $1$-bit IS model.
Each process~$i$ starts with an input for $\Pi$ which is stored in the local variable $W_i^0$ (line~\ref{itsim:input}). To limit confusion between the simulated execution of $\m{A}$, and the actual execution of Algorithm~\ref{alg:simulateIC}, we use the term \emph{round} when we talk about the simulated execution of~$\m{A}$, and \emph{iteration} for the actual execution of Algorithm~\ref{alg:simulateIC}. A round $r$ in the IC model consists for each process in writing its history in the register $M_r[i]$ and then collecting the values stored in $M_r$ to get a view $V_i^r$ (see algorithm~\ref{alg:genericIC}). The aim of the simulation is thus to provide for each process $i$ and each round $r$ a view $W_i^r$ such that $$((W_1^1,\ldots,W_n^1), (W_1^2,\ldots,W_n^2), \ldots, (W_1^r,\ldots,W_n^r),\ldots$$ could occur in some execution of the generic algorithm~\ref{alg:genericIC}. 
The simulation of a round of $\m{A}$ spans $|\m{C}^{r-1}|$ $\mathsf{write}$-$\mathsf{collect}$ iterations (line~\ref{itsim:intloop}) whose numbers coincide with the indexes of  the configurations in $\m{C}^{r-1}$. 

Recall that a configuration $c \in \m{C}^{r-1}$ is a $n$-uplet $(V_1^{r-1},\ldots,V_n^{r-1})$ of views obtained by the processes at the end of round $r$ in~$\m{A}$. In the simulation of round $r$, the goal for each process~$i\in[n]$ is to collect a (sub)set of the round-$(r-1)$ views, which will constitute its view at the end of round~$r$.
To that end, in each iteration $\rho$ such that, in the corresponding configuration~$c_\rho$, the $i$th entry is equal to the simulated view at round $r-1$ of process $i$, as stored in the local variable $W_i^{r-1}$, process~$i$ writes~$1$ in its register $M_\rho[i]$ (line~\ref{itsim:write}). A process~$j$ observing~$1$ in the $i$th entry of its snapshot of $M_\rho$ adds $W_i^{r-1}$ in its view at round $r$ (line~\ref{itsim:snap}). $W_i^{r-1}$ is deduced from the number $\rho$ of the current iteration as every process knows the indexes of the configurations in $\m{C}^r$. Hence, $W_i^r$ contains only views of round~$r-1$.

Assuming that the simulation of the first $r-1$ rounds is correct, let us consider the index~$\rho$ of the configuration that corresponds to the simulated views of round~$r-1$. A crucial observation is that, in iteration~$\rho$, each process~$i$ writes~$1$ to its register $M_\rho[i]$, and thus it does get a \emph{snapshot} of the round-$(r-1)$ views. As a consequence, its final view of round~$r$, which is the value of $W_i^r$ when the internal for-loop ends (lines~\ref{itsim:intloop}-\ref{itsim:snap}), contains this snapshot. As we shall show in Lemma~\ref{lem:validcollect}, it will follow that $(W_1^r, \ldots, W_n^r)$  may have indeed been obtained by the processes after writing and collecting their round-$(r-1)$ views, as in round~$r$ of algorithm~\m{A}.  This allows us to inductively establish the correctness of the simulation (cf.~Lemma~\ref{lem:itsim:config}). 

\paragraph{Proof of correctness.} 

Let us fix any execution $e$ of the simulation algorithm~\ref{alg:simulateIC}. Let 
\[
D_0 = (W_1^0,\ldots,W_n^0)
\]
be the inputs of the processes in~$e$, and, for any $r\in\{1,\dots,k\}$, let 
\[
D^r = (w_1^r,\ldots,w_n^r)
\] 
where $w_i^r$ is the value of the local variable $W_i^r$ at the end of simulation of round $r$, i.e, immediately after the execution of the internal for-loop of lines~\ref{itsim:intloop}-\ref{itsim:snap} corresponding of the simulation of round~$r$. We show first that the sequence $D^0,\ldots,D^k$ is a valid sequence of configurations in a possible execution of the  full-information algorithm~\m{A}.

\begin{lemma}
  \label{lem:itsim:config}
  There exists an execution $\tilde{e}$ of algorithm \m{A} in which the sequence of configurations is $(D^0,\ldots,D^k )$. 
\end{lemma}
\begin{proof}
 The proof is by induction on the (simulated) round numbers $r$. For the base case, as $D^0$ consists in the inputs of the $n$ processes it is a valid initial configuration of \m{A}. Let $r, 0 \leq r < R$ and let us assume that there is a partial $r$-round execution $\tilde{e}^r$ of \m{A} in which the sequence of global configurations is $(D_0,\ldots,D_r)$.
  Let us consider  $$D^{r+1} = (w_1^{r+1},\ldots,w_n^{r+1})$$ the configuration produced after the simulation of round $r$. Let us concentrate on some process~$i$. Suppose that $V \in w_i^r$. By the code (line~\ref{itsim:snap}), there exists an iteration $\rho$ and a process $j$ such that:
  \begin{enumerate}
  \item $\sum_{0 \leq \ell < r} |\m{C}^\ell|+1 \leq \rho \leq \sum_{0 \leq \ell \leq r} |\m{C}^\ell|$
  \item $V = c_\rho[j]$ and in iteration $\rho$, process $i$ obtains a snapshot $s_i^\rho$ with $s_i^\rho[j] = 1$. 
  \end{enumerate}
By item 2, process $j$ writes $1$ to the registers $M_\rho[j]$. Therefore, its simulated view $w_j^r$ of round $r$ is the $j$th component of the configuration $c_\rho$ of \m{A}  (line~\ref{itsim:write}). 
 Recall that $D^r = (w_1^r,\ldots,w_n^r)$ where $w_i^r$ is the value of $W_i^r$ at the end of the simulation of round $r$. By the induction hypothesis, $D^r = (V_1^r,\ldots,V_n^r)$ where $V_k^r$ is the view of process $k$ at the end of round $r$ in $\tilde{e}^r$. Hence, $V = V_j^r$, and therefore every view $V \in w_i^{r+1}$  is a view of some process at the end of the simulated execution~$\tilde{e}^r$.

 Since $D^r$ is a valid configuration at the end of round $r$ in  some execution of \m{A}, there exist $\tilde{\rho}$ such that $D^r = c_{\tilde{\rho}}$. Moreover, as $c_{\tilde{\rho}} \in \m{C}^r$ and the indexing of configurations is round-preserving, $$\sum_{0 \leq \ell < r} |\m{C}|^\ell+1 \leq \tilde{\rho} \leq \sum_{0 \leq \ell \leq r} |\m{C}|^\ell.$$ By item 1 above,  $\tilde{\rho}$ is thus  in the interval of iterations used for the simulation of round $r$.
 By the code (line~\ref{itsim:write}), the input of each $\mathsf{write}$ in this iteration is $1$. Hence, each process $i$ adds  a snapshot of the rounds $r$ views $(V_1^r,\ldots,V_n^r)$ to its set of views $W_i^{r+1}$. More precisely, in iteration $\tilde{\rho}$, let $s_1 \subset \ldots \subset s_k$ be the snapshot of $M_{\tilde{\rho}}$ obtained by the processes. Each $s_\ell$ is a $n$-component binary vector, and if process $i$ $\mathsf{snaphot}$ returns  $s_\ell$, then $s_\ell[i] = 1$.  Moreover, process $i$ then adds the set of views $$S_\ell = \{V_j^r : 1 \leq j \leq n \land s_\ell[j] =1\}$$ to $W_i^{r+1}$ and thus also to $w_i^{r+1}$. To summarize, interpreting $$w_1^{r+1},\ldots,w_n^{r+1},S_1,\ldots,S_\ell$$ as vectors, we have for each $i, 1 \leq i \leq n$ there exists $S \in \{S_1,\ldots,S_\ell\}$ such that $S \subseteq w_i^{r+1}$ and $S[i] = V_i^{r} \neq \bot$. By Lemma~\ref{lem:validcollect} (stated and proved below), there is a one round execution of the IC model 
  in which  for each process $i$ the input of  its $\mathsf{write}$ is $V_i^r$  and the output of its $\mathsf{collect}$ is~$w_i^{r+1}$.
  Therefore there exists a $r+1$-round execution of $\m{A}$ that extends execution $\tilde{e}^{r}$ and the correctness of the simulation follows. 
\end{proof}

  Next technical lemma is the missing part of the proof of Lemma~\ref{lem:itsim:config}. We essentially show that if a collection of $n$ valid snapshots can be extracted from $n$-component vectors $W_1,\ldots,W_n$, then $W_1,\ldots,W_n$ may have been returned as outputs of the $\mathsf{collect}$ operations in a round of the IC model. 

\begin{lemma} \label{lem:validcollect}
 Let $x_1, \ldots, x_n$ be $n$ values different from~$\bot$, and let $X_1,\ldots,X_n$ be $n$-dimensional vectors such that, for every $i,j\in\{1,\dots,n\}$, $X_i[j] \in \{x_j,\bot\}$. Let $S_1 \subset \ldots \subset S_m$ be $n$-dimensional vectors such that, for every $i\in\{1,\dots,n\}$, there exists $S \in \{S_1,\ldots,S_m\}$ such that 
 \[
 S[i] \neq \bot \;\mbox{and}\; S \subseteq X_i.
\]
Then there exists an execution of one round of iterated collect in which, for every process~$i\in[n]$, the input of  its $\mathsf{write}$ is $x_i$, and the output of its $\mathsf{collect}$ operation is $X_i$. 
\end{lemma}

\begin{proof}
  Let $M$ be an array of $n$ registers initialized to $[\bot,\ldots,\bot]$.  To prove the lemma, we show that there is an execution in which each process $i$ performs $\mathsf{write}(x_i)$ to $M[i]$ followed by a $\mathsf{read}$ on each component of $M[j]$ which returns $X_i[j]$. As  the specification of $\mathsf{collect}$ does not require any particular order for reading the registers, this proves that  $X_1,\ldots,X_n$ may have been returned by $\mathsf{collect}$ operations performed  by $n$ processes.
  
  We schedule the $\mathsf{write}$ operation in blocks. Let  $B_1$ consists in the set of processes $$\{j : 1 \leq j \leq n \land  S_1[j] \neq \bot \}.$$ For every $\ell\in\{2, \dots,k\}$, let   $$B_\ell=\{ j : 1 \leq j \leq n, j \notin B_{\ell-1} \land  s_\ell[j] \neq \bot\}.$$ Since $S_\ell = [x_1,\ldots,x_n]$, we have $B_1 \cup \ldots \cup B_\ell = \{1,\ldots,n\}$. We schedule first the $\mathsf{write}$ of the processes in $B_1$ in any order, followed by the  $\mathsf{write}$ of the processes in $B_2$, etc. 
  For process $i$ with output $X_i$, let $j$ such that $S_j \subseteq X_i$ and if $j\neq k$, $S_{j+1} \not\subseteq X_i$. For each $\ell$, $1 \leq \ell \leq n$, process $i$   $\mathsf{read}$ of register $M[\ell]$ occurs: 
  \begin{enumerate}
  \item  immediately after  the $j$th block write whenever $S_j[\ell] \neq \bot$ or $X_i[\ell] = \bot$, or 
  \item  immediately after the $\mathsf{write}$ by process $\ell$ to~$M[\ell]$ otherwise (i.e., whenever $S_j[\ell] = \bot$ and $X_i[\ell] \neq \bot$). Note that this $\mathsf{write}$ then occurs after the $j$th block write. 
  \end{enumerate}
After the $j$th block write, the value of register  $M[\ell]$ is $x_\ell$ if $S_j[\ell] \neq \bot$ and $\bot$ otherwise. Hence, if the $\mathsf{read}$ by process $i$ of $M[\ell]$ is scheduled  according to item 1 above, it returns $X_i[\ell]$. Clearly, if it is scheduled according to item 2, it returns also $X_i[\ell] = x_\ell$. It remains to show that every $\mathsf{read}$ by process $i$ is performed after its $\mathsf{write}$ to $M[i]$. According to item 1 and 2, each $\mathsf{read}$ by process $i$ occurs after the $j$th block write. For some $S \in \{S_1,\ldots,S_k\}$, we have 
$S[i] \neq \bot$ and $S \subseteq X_i$. Since $S_j$ is the largest vector in $\{S_1,\ldots,S_k\}$ contained in $X_i$, $S \subseteq S_j$. Hence, $S_j[i] \neq \bot$ and therefore the $\mathsf{write}$ to $M[i]$ occurs in the $j$th block  or in a previous block. Since all $\mathsf{read}$ by process $i$ are scheduled after the $j$th block write, they  occur after process $i$ has written to~$M[i]$. 
\end{proof}

This completes the proof of Proposition~\ref{prop:IC2IIS-1-bit}.

\subsection{From IS with Unbounded Registers to IC with Unbounded Registers}
\label{sec:IC_to_IS}

To complete the proof of Theorem~\ref{theo:iterated-immediate-snapshot}, we now show that the iterated collect and the iterated snapshot models have the same computational power. Specifically, we show that a round of the IS model can be simulated using $O(n)$ iterations of the IC model. We reuse an existing shared memory snapshot algorithm and slightly adapt it to the iterated case.

\begin{proposition}
  \label{prop:simIS_in_IC}
  For every task $\Pi$, if $\Pi$ is solvable in the IS model (with unbounded registers) then  $\Pi$ is solvable in the IC model (with unbounded registers). 
\end{proposition}

\begin{proof}
  We simulate (Algorithm~\ref{alg:unboundedICsimIS}) one round of the IS model with $n$ $\mathsf{write}$/$\mathsf{collect}$ iterations using the  (immediate) snapshot algorithm presented by Borowsky and Gafni in \cite{BG92}. 
Let  $x_i$ denote  the input of process $i$ to its  $\mathsf{write}$ operation in the simulated round of the IS model. In each iteration   $\rho$, process $i$ writes its input $x_i$ and a Boolean indicating whether or not it has already managed to compute  a valid snapshot (line~\ref{issim:write_collect}). To obtain a snapshot at iteration $\rho$, process $i$ must have collected in this iteration  $n+1 -\rho$ inputs $I = \{x_{i1},\ldots,x_{i(n+1 - \rho)}\}$ of processes that have not yet obtained a snapshot (line~\ref{issim:if}). If this occurs, the vector corresponding to the set $I$ is the output for process $i$ of the simulated $\mathsf{snapshot}()$  operation (line~\ref{issim:snap}). 

\begin{algorithm}[tb]
  \caption{Borowsky-Gafni's snapshot algorithm adapted to the IC model for simulating one round of the IS model. Code for process $i$. $x_i$ is the input of $i$ for the simulated round.}
  \label{alg:unboundedICsimIS}

  \begin{algorithmic}[1]
    \Statex \textbf{shared variables} $M_1,\ldots,M_n$ arrays of $n$ SWMR registers, initially $[\bot,\ldots,\bot]$ 
    \Statex \textbf{local variables} 
    $S_i$, output of simulated $\mathsf{snapshot}()$ initially $[\bot,\ldots,\bot]$
\Statex    $~~~~~~~~~~~~~~~~~~~~~$ $b_i$ Boolean initially false  
    \Begin\label{issim:begin}
    \For{$\rho=1,\ldots,n$}\label{issim:for}
    \State\label{issim:write_collect}$M_\rho[i].\mathsf{write}(x_i,b_i)$;  
    \State $X_i \gets M_\rho.\mathsf{collect}()$
    \If{$b_i=false$ \textbf{and} $|\{ j\in [n] : X_i[j] = (x_j,\emph{false})\}| = n+1 - \rho$}\label{issim:if}
    	\For{$j=1,\dots,n$}\label{issim:snap}
     		\If{$X_i[j] = (x_j,false)$}
			\State $S_i[j] \gets x_j$
		\Else
			\State $S_i[j] \gets \bot$
		\EndIf
     	\EndFor
      	\State $b_i \gets \emph{true}$
    \EndIf
    \EndFor
    \State\label{issim:return} \textsl{return} $S_i$  \Comment{output of the simulated $\mathsf{snapshot}$ operation}
    \End
  \end{algorithmic}
\end{algorithm}

We briefly sketch the correctness proof of the simulation. 
Inductively, at the beginning  of each iteration $\rho$ the number of processes without snapshot is  at most $n+1-\rho$. Indeed, in  the first iteration  the collect $X_\ell$ of  process $\ell$ that last $\mathsf{write}s$ to $M_1$ contains the input of every process. Therefore at least one process obtains a snapshot in the first iteration. Similarly, if at the beginning of iteration $\rho \geq 2$, a set $U^\rho$ of  $n+1 - \rho$ processes is without  snapshot, the  process in $U^\rho$  that last $\mathsf{write}$ to $M_\rho$ sees   the inputs of every process in $U^\rho$ and therefore obtains a valid snapshot. It thus follows that every process $i$ manages to obtain a snapshot by the end of iteration $n$.

Suppose that process $i$ obtains a snapshot $S_i$ (line~\ref{issim:snap}) in iteration $\rho$.  $S_i$ has exactly $n+1-\rho$ non-$\bot$ entry, and each of theses entries  corresponds to a process in the set $U^\rho$ of processes that have not yet obtained a snapshot at the beginning of iteration $\rho$. Clearly, $S_i[i] \neq \bot$.  As shown above $|U^\rho| \leq n+1 - \rho$. Therefore, the snapshots obtained by the processes  in iteration $\rho$ are the same.  Finally, if process $i'$ obtains its snapshot in iteration $\rho' > \rho$, $S_{i'} \subset S_i$ since the non-$\bot$ entries of $S_{i'}$ are those corresponding to the set of processes without snapshot at the beginning of $\rho'$, which is a subset of $U^\rho$. 
\end{proof}



\section{Faster 2-processes Wait-Free $\epsilon$-Agreement}
\label{sec:faster-eps-agr}

When there is no restriction on the size of the registers, 2-processes wait-free binary $\epsilon$-agreement is known to be solvable in  $O(\log \frac{1}{\epsilon})$ steps per process, and this bound is best possible (see, e.g.~\cite{HoestS06}).  Algorithm~\ref{alg:epsilon} uses 1-bit registers with a novel technique, but  is exponentially slower, with  $O(\frac{1}{\epsilon})$ steps per process. 
In this section, we show that this slowdown is not inherent to the fact that each  register has constant size. We present a more involved solution performing in $O(\log \frac{1}{\epsilon})$ steps, for two processes using registers of constant size.

The final states of a protocol $\m{P}$ solving a task can be arranged in a graph $G(\m{P})$. Each vertex  represents a final local state of some process and an edge connects two vertexes whenever the corresponding local states coexist at the end of some execution of the protocol.   Vertexes may be \emph{colored}  process indexes:  if vertex $v$ represents the final state of process $p$, its color $id(v)$ is the index of $p$. An important fact about wait-free models is that the protocol graph is connected. Specifically, for two processes and fixed inputs (for example, in the case of  $\epsilon$-agreement, the input of process $1$ is $0$ and the input of process $2$ is $1$), it is  known \cite[chapter 2]{bookHerlihyKR2013} that the protocol graph is a \emph{chromatic} path. Each extremity corresponds to a \emph{solo} execution, in which one process completes the protocol without the other process taking step. An internal vertex $v$ with $id(v)= i$ is connected to two vertexes $u,u'$ with $id'(u) = id(u') = j \neq i$. See Figure~\ref{fig:ISexe} for an example in the IS model. The ability of a protocol to solve $\epsilon$-agreement  depends on the length of the path: the longer the path is, the smaller $\epsilon$ can be chosen. Indeed, in a solo execution process $i$ with input $x \in \{0,1\}$ has to decide $x$ and therefore $\epsilon$ cannot be smaller than $\frac{1}{L}$, where $L$ is the length of the path.

How fast the length of the path grows as processes take more and more steps? In the IS model for two processes, with no restriction on the size of the register, the length is $3^{r}$ after $r$ rounds as we explain next. 
Recall that in the iterated snapshot model, executions proceed in asynchronous rounds.
In each round $r$,  using a fresh array of shared registers $M_r$, each process writes a value to its register and takes a snapshot of the array. For two processes, taking a snapshot is the same as reading both registers. In each round, they are three possibilities: one of the two processes is \emph{solo}, that is, sees only its own value (the other process thus sees both values.), or the two processes see each other's value. Hence, the number of possible executions grows exponentially with the round numbers: there are $3^r$ executions after $r$ rounds.  Binary $\epsilon$-agreement can be solved in $r$ rounds, for any $\epsilon \geq \frac{1}{3^r}$ by assigning the values $0,\frac{1}{3^r},\frac{2}{3^r}, \ldots, 1$ to the vertexes representing the local state after $r$ rounds, starting from the vertex representing the state of the process with input $0$ running solo.

\begin{figure}[tb]
  \centering
\newcommand{\level}{1}
\newcommand{\scale}{1.8}
 \begin{tikzpicture}[
   Astyle/.style = {color=black, fill= black, shape=circle, minimum size=5},
   Bstyle/.style = {color=red, fill= red, shape=circle,minimum size=5},
   Aedge/.style = {dashed, black, thick},
   Bedge/.style = {dashed, red, thick}]

   \foreach \r in {0,1,2,3}{
     \pgfmathsetmacro\y{-\r}
     \tkzDefPoints{-5.4/\y/D\r, 5.4/\y/F\r}
     \tkzDrawLine[color=gray,thin,dashed](D\r,F\r)
     \tkzLabelLine[gray,above right,pos=1.1](D\r,F\r){\small{round $\r$}}
     \pgfmathsetmacro\xA{-(1*\scale^\r)}
     \pgfmathsetmacro\xB{1*\scale^\r}
     \tkzDefPoints{\xA/\y/A\r, \xB/\y/B\r}
     \draw (A\r) -- (B\r);

     \ifthenelse{\r=0}{ 
       \tkzDrawPoint[Astyle](A\r)
       \tkzDrawPoint[Bstyle](B\r)
     }
     {
       \tkzCalcLength(A\r,B\r)\tkzGetLength{ABl} 
       \renewcommand{\level}{\r}
       \pgfmathsetmacro\nx{int(3^\level)}
       \foreach \x in {0,2,...,\nx}{
         \pgfmathsetmacro\mypos{\x/(\nx)}
         \tkzDefPointOnLine[pos=\mypos](A\r,B\r) \tkzGetPoint{A\r\x}
         \tkzDrawPoint[Astyle](A\r\x)
       }
       \foreach \x in {1,3,...,\nx}{
         \pgfmathsetmacro\mypos{\x/\nx}
         \tkzDefPointOnLine[pos=\mypos](A\r,B\r) \tkzGetPoint{B\r\x}
         \tkzDrawPoint[Bstyle](B\r\x)
       }
     }
   }
   \draw[Aedge] (A0) -- (A10);    \draw[Aedge] (A0) -- (A12);
   \draw[Bedge] (B0) -- (B11);    \draw[Bedge] (B0) -- (B13);
   \foreach \r/\R in {1/2, 2/3}{
     \renewcommand{\level}{\r}
     \pgfmathsetmacro\nx{int(3^\level)}
     \foreach \x in {0,2,...,\nx}{
       \pgfmathsetmacro\childsolo{int(3*\x)}
       \pgfmathsetmacro\childleft{int(3*\x-2)}
       \pgfmathsetmacro\childright{int(3*\x+2)}
       \ifthenelse{\x>0}{\draw[Aedge] (A\r\x) -- (A\R\childleft);}{}
       \draw[Aedge] (A\r\x) -- (A\R\childsolo);
       \draw[Aedge] (A\r\x) -- (A\R\childright);
     }
     \foreach \x in {1,3,...,\nx}{
       \pgfmathsetmacro\childsolo{int(3*\x)}
       \pgfmathsetmacro\childleft{int(3*\x-2)}
       \pgfmathsetmacro\childright{int(3*\x+2)}
       \draw[Bedge] (B\r\x) -- (B\R\childleft);
       \draw[Bedge] (B\r\x) -- (B\R\childsolo);
       \ifthenelse{\x<\nx}{\draw[Bedge] (B\r\x) -- (B\R\childright);}{}
     }
   }
 \end{tikzpicture}
 \label{fig:ISexe}
  \caption{\sl 3-rounds executions of the iterated snapshot model for two processes. Each black (resp. red) dot represents a local state of process $0$ (resp. process $1$). After one round, the left edge represents the execution in which process $0$ is solo and process $1$ reads the registers after process $1$ has written. In the execution represented by the middle edge, both processes read each other value. Process $1$ cannot known the view of process $0$ and has thus  the same state in both executions. The edge on the right represent the symmetric execution: process $1$ is solo and process $0$ reads the registers after process $1$ has written. The same pattern repeats every round: each edge is subdivided in three edges.}
\end{figure}
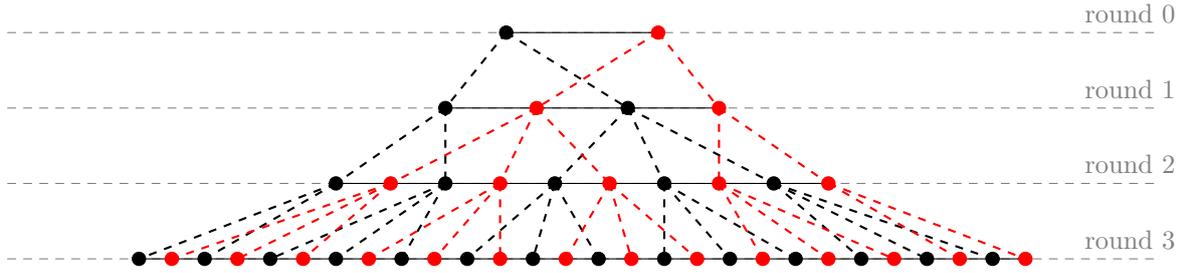

\subsection{Labelling protocols}

Even if the size of the registers is unbounded, and the protocol is full information, some executions remain  indistinguishable to the processes. For example, when a process in the last round $r$ sees both values, it has no way to tell whether the other process is solo or also sees both values. What happens if we restrict the size of the registers? Are there more executions that become indistinguishable to the processes? That is, if processes are allowed to communicate only a constant number of bit per round, as in Algorithm~\ref{alg:epsilon}, how many are there distinct final states after $r$ rounds of computation?

At one  extreme, if each process always writes the same value (say $1$) in each round, process $i$ will end up in the same state in every  execution in which it sees the other process' value in each round. This is not the case in full information protocols where process $j \neq i$ can communicate whether it was solo or not in the previous rounds.  Perhaps surprisingly, it is shown in~\cite{Delporte-Gallet20} that each process writing only one bit of information in each round is sufficient for the set of local  states after $r$ rounds to be of size $3^r+1$.

More precisely, a \emph{labelling} protocol is an IS protocol that has no input. It outputs at each process $i$ and at the end of each round $r$ a \emph{label} of the form $(i,r,\lambda)$ where $\lambda$ is taken from  some domain~$\m{L}$. 

\begin{lemma}[\cite{Delporte-Gallet20}]
  \label{lem:1bit_disambiguation}
  There exists a two processes deterministic labelling protocol for the IS model in which, in each round $r$, each process $i$ writes one bit.  
  The size of the set of labels output at  round $r$ in all possible executions is $3^r+1$. 
\end{lemma}

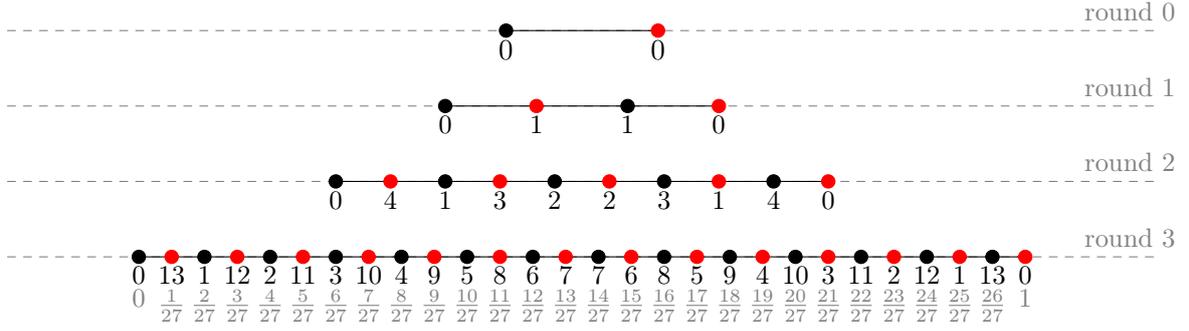
\begin{figure}[tb]
  \newcommand{\level}{1}
  \newcommand{\scale}{1.8}
  \centering
   \begin{tikzpicture}[
   Astyle/.style = {color=black, fill= black, shape=circle, minimum size=5},
   Bstyle/.style = {color=red, fill= red, shape=circle,minimum size=5},
   Aedge/.style = {dashed, black, thick},
   Bedge/.style = {dashed, red, thick}
   ]

   \foreach \r in {0,1,2,3}{
     \pgfmathsetmacro\y{-\r}
     \tkzDefPoints{-5.4/\y/D\r, 5.4/\y/F\r}
     \tkzDrawLine[color=gray,thin,dashed](D\r,F\r)
     \tkzLabelLine[gray,above right,pos=1.1](D\r,F\r){\small{round $\r$}}
     \pgfmathsetmacro\xA{-(1*\scale^\r)}
     \pgfmathsetmacro\xB{1*\scale^\r}
     \tkzDefPoints{\xA/\y/A\r, \xB/\y/B\r}
     \draw (A\r) -- (B\r);

     \ifthenelse{\r=0}{ 
       \tkzDrawPoint[Astyle](A\r)
       \tkzLabelPoint(A\r){0}
       \tkzDrawPoint[Bstyle](B\r)
       \tkzLabelPoint(B\r){0}
     }
     { 
       \tkzCalcLength(A\r,B\r)\tkzGetLength{ABl} 
       \renewcommand{\level}{\r}
       \pgfmathsetmacro\nx{3^\level}
       \foreach \x in {0,2,...,\nx}{
         \pgfmathsetmacro\mypos{\x/(\nx)}
         \tkzDefPointOnLine[pos=\mypos](A\r,B\r) \tkzGetPoint{A\r\x}
         \pgfmathparse{div(\x,2)}
         \tkzDrawPoint[Astyle](A\r\x)
         \tkzLabelPoint(A\r\x){\small{\pgfmathresult}}
       }
       \foreach \x in {1,3,...,\nx}{
         \pgfmathsetmacro\mypos{\x/\nx}
         \tkzDefPointOnLine[pos=\mypos](A\r,B\r) \tkzGetPoint{B\r\x}
         \pgfmathparse{div(\nx-\x,2)}
         \tkzDrawPoint[Bstyle](B\r\x)
         \tkzLabelPoint(B\r\x){\small{\pgfmathresult}}
       }
     }
     }
     \tkzLabelPoint[color=gray,below=.3cm](A30){\small{$0$}}
     \foreach \x in {2,4,...,27}{
       \tkzLabelPoint[color=gray,below=.3cm](A3\x){\small{$\frac{\x}{27}$}}
     }
     \tkzLabelPoint[color=gray,below=.3cm](B327){\small{$1$}}
     \foreach \x in {1,3,...,26}{
       \tkzLabelPoint[color=gray,below=.3cm](B3\x){\small{$\frac{\x}{27}$}}
     }
    \end{tikzpicture}
    \caption{\sl The labels (black) produced by the protocol in~\cite{Delporte-Gallet20}. Round numbers and process indexes are omitted. Labels are associated with values (grey) in $[0,1]$ for solving $\epsilon$-agreement.}
    \label{fig:labels}
  \end{figure}

  \paragraph{From labelling to $\epsilon$-agreement.}
  
It follows from Lemma~\ref{lem:1bit_disambiguation} that $\epsilon$-agreement can be solved in $O(\log\frac{1}{\epsilon})$ rounds by 2 processes in the IS model with $2$-bit registers: run the labelling protocol for $r$ rounds, where $\frac{1}{3^r} \leq \epsilon$ and associates with each round $r$ label $\lambda$ a value $f(\lambda) \in \{0,\frac{1}{3^r}, \frac{2}{3^r}, \ldots, 1\}$ such that:
\begin{itemize}
\item $f(\lambda_{s0}) = 0$, where $\lambda_{s0}$ is the label obtained by process $0$ in the execution in which it is solo in every round. Similarly, for the label $\lambda_{s1}$ obtained by process $1$ when it is solo in every round, $f(\lambda_{s1}) = 1$. 
\item If there is an $r$-rounds execution  in which both labels $\lambda$ and $\lambda'$ are obtained by processes $0$ and $1$, $|f(\lambda) - f(\lambda')| = \frac{1}{3^{r}}$
\end{itemize}
See Figure~\ref{fig:labels} for an example for the case  $r=3$ and $\epsilon = \frac{1}{27}$.

For communicating its binary input $x_i$, each process in each round writes $x_i$ and the bit of the labelling protocol.  
After obtaining a label $\lambda_i$, if it has discovered that both inputs are the same, or it has not seen the input of the other process, it decides its own input $x_i$. Otherwise, process $i$ decides $y_i$: $$ \mbox{if $\lambda_i < \frac{1}{2}$ then}~ y_i = \left\{
  \begin{array}{l}
    f(\lambda_i) ~\mbox{if $x_0 = 0$} \\
    1 - f(\lambda_i) ~\mbox{otherwise}
  \end{array}\right.
~\mbox{else}~ y_i = \left\{
  \begin{array}{l}
    f(\lambda_i) ~\mbox{if $x_1 = 1$} \\
    1 - f(\lambda_i) ~\mbox{otherwise}
  \end{array}\right.$$
We briefly check that agreement and validity are satisfied. Suppose that a process, say process~$0$ does not see the input of the other process. Then it was solo in each round of the IS execution and consequently $f(\lambda_0) = 0$ and thus $f(\lambda_1) \leq \frac{1}{\epsilon}$.
Process $0$ decides $x_0$. If $x_1 = x_0$, process $1$ also decides $x_0$. Otherwise, since process $1$ sees $x_0$ it decides $f(\lambda_1)$ if $x_0 = 0$ or $1-f(\lambda_1)$ if $x_0 = 1$. Both decisions are thus at most $\frac{1}{\epsilon}$ apart, and in the range $[\min(x_0,x_1),\max(x_0,x_1)]$ as required. The case in which process $1$ does not see the input of process $0$ is similar.

Suppose that both processes see each other input. If $x_0 = x_1$ they decide the same value $x_0 = x_1$. If $x_0 \neq x_1$, the values decided are at most $\frac{1}{\epsilon}$ apart if both $f(\lambda_0),f(\lambda_1) < \frac{1}{2}$, or both $f(\lambda_0),f(\lambda_1) \geq \frac{1}{2}$. Let us assume that $f(\lambda_0) < \frac{1}{2}$ and  $f(\lambda_1) \geq \frac{1}{2}$ (the other case is similar). Then 
\begin{align*}
|y_0 - y_1| & = | x_0 + (-1)^{x_0}f(\lambda_0) - (1-x_1) - (-1)^{1-x_1}f(\lambda_1)| \\
& = |(x_0+x_1) -1 \pm{} (f(\lambda_0) - f(\lambda_1))| \\
& = |f(\lambda_0) - f(\lambda_1)|\\
& \leq \frac{1}{\epsilon},
\end{align*}
as required. 

If input registers are available, processes do not need to communicate their input in each round. They write their input to the input registers, execute the IS labelling protocol, read the input register and then decide as explained above. We thus have:

\begin{lemma}
  \label{lem:lcis:fromlabeltoeps}
  Let $f: \mathbb{N} \to \mathbb{N}$ be a growing function and let $\m{P}$ be a labelling protocol for two processes that generates in total $f(r)$ labels in all executions in which processes take at most $r$ steps. If input registers are available, $\epsilon$-agreement is solvable with step complexity $O(r)$ for any $\epsilon \geq \frac{1}{f(r)}$. 
\end{lemma}

Nevertheless, applying the lemma to the labelling protocol $\m{P}$ of~\cite{Delporte-Gallet20} requires
registers of size $O(\log \frac{1}{\epsilon})$ as in the iterated snapshot model a new array of shared registers is used in each round. We show next how  to simulate a subset of the executions of $\m{P}$ using two constant size registers. 

\subsection{Simulating IS execution with constant size registers}

It is well known that write-snapshot objects can be implemented with SWMR registers~\cite{AfekADGMS93},  and hence, given two unbounded registers two processes can wait-free simulate  IS executions. 

Unbounded space is required, as one process $p$ may be simulating some round $R$ while the other process $q$ may still be simulating some lower round $r$. The difference $R-r$ cannot be bounded, and  enough information has to be kept in shared memory for $q$ to be able to simulate rounds $r+1, \ldots,R$. This is because $q$ cannot ask and wait for $p$ to tell him the values it has written in  rounds $r+1, \ldots, R$, or at least  how many consecutive  solo rounds it has performed. We get around this issue by simulating only a \emph{subset} of  all possible IS executions in a \emph{non-uniform} way. Essentially,  after a process has simulated $\Delta$ consecutive solo rounds, it quits the simulation, where $\Delta$ is a parameter of the simulation.  More precisely, let $R$ be the maximal number of  rounds to be simulated. There are some IS executions of length $R$ that are never simulated (the ones in which a processes is solo for more than $\Delta$ rounds), and the processes may  perform less that $R$ rounds. Nevertheless,  we show that the total number of simulated IS executions is $\Omega(2^R)$ for $\Delta \geq 2$. See Figure~\ref{fig:sim_exe} for the case $R=5$ and $\Delta =2$.

Therefore, applying the simulation to a labelling protocol, such as the one in~\cite{Delporte-Gallet20}, $\Omega(2^R)$ distinct labels are globally (over all possible runs of the simulation) generated. It thus  follows that processes can solved $\epsilon$-agreement with $\epsilon \geq \frac{1}{2^R}$. As simulating one IS round takes a constant number of writes and reads, the step complexity is $O(R) = O(\log\frac{1}{\epsilon})$, while the size of the shared registers remains constant.

\begin{figure}[tb]
  \centering
\newcommand{\level}{1}
\newcommand{\scale}{1.6}
\newcommand{\drawpoint}[2]{
    \ifthenelse{\isodd{#1}}{ 
      \tkzDrawPoint[Bstyle](B#2#1)
    }{ 
      \tkzDrawPoint[Astyle](A#2#1)
    }
  }
\newcommand{\drawsoloedge}[2]{
  \pgfmathsetmacro\childsolo{int(3*#1)}
  \pgfmathsetmacro\nextlevel{int(#2+1)}
  \ifthenelse{\isodd{#1}}{ 
    \draw[SBedge] (B#2#1) -- (B\nextlevel\childsolo);
  }
  { 
    \draw[SAedge] (A#2#1) -- (A\nextlevel\childsolo);
  }
}

\newcommand{\drawsoloedgesandpoints}[3]{
  \drawsoloedges{#1}{#2}{#3}
  \pgfmathsetmacro\nextlevel{int(#3+1)}
  \pgfmathsetmacro\firstindex{int(#1*3)}
  \pgfmathsetmacro\lastindex{int(#2*3)}
  \drawpoints{\firstindex}{\lastindex}{\nextlevel}
  }
\newcommand{\drawsoloedges}[3]{
  \pgfmathsetmacro\secondindex{int(#1+1)}
  \ifthenelse{#1=\secondindex}{
    \drawsoloedge{#1}{#3}
    \drawsoloedge{\secondindex}{#3}
  }
  {
    \foreach \i in {#1,\secondindex,...,#2} {
      \drawsoloedge{\i}{#3}
    }
  }
  }
\newcommand{\drawpoints}[3]{
    \pgfmathsetmacro\secondindex{int(#1+1)}
    \ifthenelse{#1=\secondindex}{
      \drawpoint{#1}{#3}
      \drawpoint{\secondindex}{#3}
    }
    {
      \foreach \i in {#1,\secondindex,...,#2} {
        \drawpoint{\i}{#3}
      }
    }
}
\newcommand{\suppress}[3]{
  \foreach \d in {-2,-1,0,1,2} {
    \pgfmathsetmacro\i{int(#1+\d)}
    \ifthenelse{\i < 0 \OR \i > #3}{
     }{
    \ifthenelse{\isodd{\i}}{ 
      \tkzDrawPoint[Gstyle](B#2\i)
    }
    { 
      \tkzDrawPoint[Gstyle](A#2\i)
    }
    }
  }
}
\hspace*{-2.7cm}\begin{tikzpicture}[
   Astyle/.style = {color=black, fill= black, shape=circle, minimum size=5},
   Bstyle/.style = {color=red,  fill= red, shape=circle,minimum size=5},
   Gstyle/.style = {color=gray, fill= gray, shape=circle,minimum size=5},
   Aedge/.style = {dashed, black, thick},
   Bedge/.style = {dashed, red, thick},
   SAedge/.style = {Aedge},
   SBedge/.style = {Bedge}]

   \foreach \r in {0,1,2,3,4,5}{
     \pgfmathsetmacro\y{-\r}
     \tkzDefPoints{-5.4/\y/D\r, 5.4/\y/F\r}
     \tkzDrawLine[color=gray,thin,dashed](D\r,F\r)
     \tkzLabelLine[gray,above right,pos=1.1](D\r,F\r){\small{round $\r$}}
     \pgfmathsetmacro\xA{-(1*\scale^\r)}
     \pgfmathsetmacro\xB{1*\scale^\r}
     \tkzDefPoints{\xA/\y/A\r, \xB/\y/B\r}

     \ifthenelse{\r=0}{ 
       \tkzDrawPoint[Astyle](A\r)
       \tkzDrawPoint[Bstyle](B\r)
     }
     {
       \tkzCalcLength(A\r,B\r)\tkzGetLength{ABl} 
       \renewcommand{\level}{\r}
       \pgfmathsetmacro\nx{3^\level}
       \foreach \x in {0,2,...,\nx}{
         \pgfmathsetmacro\mypos{\x/(\nx)}
         \tkzDefPointOnLine[pos=\mypos](A\r,B\r) \tkzGetPoint{A\r\x}
       }
       \foreach \x in {1,3,...,\nx}{
         \pgfmathsetmacro\mypos{\x/\nx}
         \tkzDefPointOnLine[pos=\mypos](A\r,B\r) \tkzGetPoint{B\r\x}
       }
     }
   }
    \draw (A0) -- (B0);
    \draw (A1) -- (B1);
    \draw (A2) -- (B2);
    \draw (B33) -- (A324);
    \draw (B49) -- (A424);     \draw (A430) -- (B451);     \draw (B457) -- (A472);   
    \draw[SAedge] (A0) -- (A10);    
    \draw[SBedge] (B0) -- (B13);
    \drawsoloedges{0}{3}{1}
    \drawsoloedges{1}{8}{2}
    \drawsoloedges{3}{8}{3} \drawsoloedges{10}{17}{3} \drawsoloedges{19}{24}{3}
    \drawsoloedgesandpoints{10}{17}{4} 
    \drawsoloedgesandpoints{19}{24}{4} 
    \drawsoloedgesandpoints{30}{35}{4}
    \drawsoloedgesandpoints{37}{44}{4}
    \drawsoloedgesandpoints{46}{51}{4}
    \drawsoloedgesandpoints{57}{62}{4}
    \drawsoloedgesandpoints{64}{71}{4}
    \drawpoints{0}{3}{1}
    \drawpoints{0}{9}{2}
    \drawpoints{3}{24}{3}
    \drawpoints{9}{24}{4}\drawpoints{30}{51}{4}\drawpoints{57}{72}{4}

    \end{tikzpicture}
    \caption{\sl Simulated subset of IS. Each process exits the simulation after simulating $\Delta=2$ consecutive solo-round (dashed-lines.) or after simulating $R=5$ rounds.  The number of simulated IS executions still grows exponentially with  $R$.}
    \label{fig:sim_exe}
\end{figure}
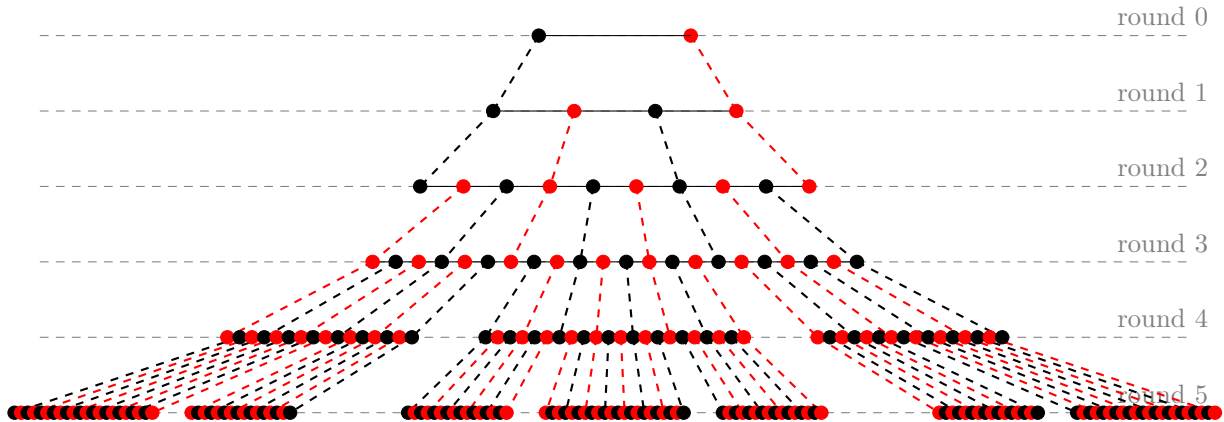

\paragraph{A simulation}
Let $\m{P}$ be a labelling protocol, for example the one described in~\cite{Delporte-Gallet20}. $\m{P}$ is a round based protocol designed for the IS model. For each process $i$, it consists in two functions:
\begin{itemize}
\item \textsc{write}($r,i, view_1,\ldots,view_{r-1}$) that provides the value process $i$ writes in round $r$. It depends on the current round number $r$, the index $i$ of the process and the successive snapshots (or \emph{views}) obtained by the process in the previous rounds.  
\item \textsc{label}($r,i,view_1,\ldots,view_{r}$) that gives the label at the end of round $r$ for process $i$. Here also, the label depends on the local state of the process, that is its index together with the sequence of views obtained in rounds $1,\ldots,r$. 
\end{itemize}
In Algorithm~\ref{alg:constantISsim}, we simulate executions of $\m{P}$ using two registers $R_1, R_2$ of constant size. The constant depends on how many bits $\m{P}$ writes in each round, and on a parameter $\Delta$. In each execution of Algorithm~\ref{alg:constantISsim}, an IS execution $E$ of at most $R$ rounds is simulated, and when the simulation terminates at process $i$, the label process $i$ would have received at the end of $E$ is returned (line~\ref{lcis:returnlabel}). The simulated execution is \emph{non-uniform}, in the sense that processes $1$ and $2$ may have performed different numbers of rounds in $E$. $R$, the maximum number of rounds performed in $E$, may be arbitrarily large. This has no impact on the size of the registers $R_1,R_2$.

An execution in the IS model generates at each process $i$ a sequence of \emph{views} $(view^1_i,view^2_i,\ldots)$, one per round.  $view_i^r$ is a $2$-components vector, result of the  snapshot performed by process $i$ on the $r$th memory $M_r$. $view_i^r[i] = v$, where $v$ is the value the simulated protocol writes at round~$r$ (1 bit in case of the labelling protocol of~\cite{Delporte-Gallet20}) (line~\ref{lcis:seeme}). For the  index $j$ corresponding to the other process, $view^r[j]$ is $\bot$ if $i$ is solo in the simulated round $r$ (line~\ref{lcis:solo}), or the value written by $j$ according to the  simulated protocol in round $r$ (line~\ref{lcis:notsolo}).

In the simulated IS execution, each process is solo for at most $\Delta$ consecutive rounds.
The local variable $c$ of process $i$ counts the number of consecutive simulated solo rounds for process $i$ (line~\ref{lcis:solo}). If this counter reaches $\Delta$ (line~\ref{lcis:simend}), process $i$ quits the simulation. The algorithm hence maintains the invariant that  processes are at the same time simulating  IS rounds whose numbers are at most $\Delta$ apart. Therefore, it is sufficient for process $i$ to keep an history $H$ of the values the simulated protocol writes in the last $\Delta+1$ rounds. When process $i$ starts the simulation of a new round $r$, it discards the value it wrote in the simulated round $r-\Delta$, replacing it with the new value $v$ to be written in the current round (line~\ref{lcis:prepareH}). The history $H$ is then written in the register $R_i$ of process $i$, which allows process $j\neq i$ to complete its view if it is simulating a round $r' =r - \delta$ with $0 \leq \delta \leq \Delta$. The value written by process $i$ in round $r'$ is stored at position $\delta$ in $H$.

Note that the processes cannot communicate the round number they  are currently simulating, since this would require a non constant number of bits. Instead, processes communicate their position in a ring of  size $2\Delta+1$. The ring is oriented and its nodes are numbered from $0$ to $2\Delta$. Initially, both processes are in node $0$. Each time a process simulates a new round, it moves to the next node (line~\ref{lcis:preparex}) and writes its new position $x$ in its register (line~\ref{lcis:write}). The last round  simulated by the other process $j$ can  be deduced  from its position on the ring. Let  $x_j$ be the current position of $j$ on the ring (as seen by process $i$ when it reads $j$'s register) and $xprec_j$ be its last  position known by $i$, and $d = \ell(xprec_j,x_j)$ the smallest number of moves required to go from node $xprec_j$ to node $x_j$ (the length of the directed path from $xprec_q$ to $x_j$). Process $i$ assumes that $j$ has simulated $d$ rounds since the last time it read $j$'s register. This might be incorrect, as process $i$ observe the same position for $j$ if $j$ has performed one or more complete lap, e.g., $j$ has simulated $d + k(2\Delta+1)$ rounds, for some $k>0$. This is not possible, since, as we show, it would require  process $j$ to simulate  more than $\Delta$ consecutive solo rounds. Hence, process $i$ is able to accurately computes the last round $r_j$ simulated by process $j$. The history $H_j$ it reads from process $j$ register (line~\ref{lcis:read}) thus contains the values written in the simulated execution by process $j$ at round $r_j,r_j-1,\ldots,r_j-\Delta$. 

When the simulation ends at process $i$ (either because $\Delta$ consecutive rounds in which $i$ is solo have been simulated, or process $i$ has simulated $R$ rounds), process $i$ returns the label generated by the  labelling  protocol $\m{P}$ at the end of the simulated execution (line~\ref{lcis:returnlabel}). The register $R_i$ should be large enough to contains the history $H$ and a position on the ring. In the labelling protocol of~\cite{Delporte-Gallet20},  one bit is written in each round. Hence, registers of size $O(\Delta + \log\Delta)$ suffice. 

\begin{algorithm}[tb]
  \caption{Simulation of an IS labelling  protocol $\m{P}$. Code for process $me \in \{1,2\}$.}
  \label{alg:constantISsim}
  \begin{algorithmic}[1]
    \Statex\textbf{local variables}
    \Statex\hspace{\algorithmicindent} $other \gets 2-me$ \Comment{index of the other process}
    \Statex\hspace{\algorithmicindent} $estr \gets 0$ \Comment{estimation of the  current round of  proc. $other$}
    \Statex\hspace{\algorithmicindent} $xprec_{other} \gets 0$ \Comment{last known position of the proc. $other$ on the ring}
    \Statex\hspace{\algorithmicindent} $c \gets 0$ \Comment{number of consecutive simulated solo rounds by $me$}
    \Statex\hspace{\algorithmicindent} $H \gets [\bot,\ldots,\bot]$ 
    \Comment{history of the last $\Delta+1$ values written by $me$}
    \Statex\textbf{shared variables}
    \Statex\hspace{\algorithmicindent} $R_k \gets (0,[\bot,\ldots,\bot])$ for $k \in \{1,2\}$ \Comment{shared registers.}
\Statex    \Comment{$0$ is the initial position on the ring, $[\bot,\ldots,\bot]$ the initial history for both process}
    \Begin\label{lcis:begin}
    \For{$r = 1,\ldots,R$}\label{lcis:for}
    	\State\label{lcis:preparex}$x \gets r \mod (2\Delta+1)$; \Comment{advance one step in the ring}
   	 \State\label{lcis:Pwrite}$v \gets \textsc{write}(r,me,view_1,\ldots,view_{r-1})$ 
    		\Comment{value written in round $r$ by proc. $me$ in $\m{P}$}
    	\State \textbf{for} every $j \in \{1,\ldots,\Delta\}$ \textbf{do} $H[j] \gets H[j-1]$ \label{lcis:prepareH}
    	\State $H[0] \gets v$ \Comment{discard oldest value from history $H$ and store value of round $r$}
    	\State\label{lcis:seeme}$view^{r}[me] \gets v$ \Comment{proc. $me$ always sees its own value in each round}  
    	\State\label{lcis:write}$R_{me}.\mathsf{write}(x,H)$  
   	 \State\label{lcis:read}$(x_{other}, H_{other}) \gets R_{other}.\mathsf{read}()$;
    	\State\label{lcis:estimateround}$estr \gets estr + \ell(xprec_{other},x_{other})$ \Comment{$\ell(xprec_{other},x_{other})$ is the length of the...}
    	\State $xprec_{other} \gets x_{other}$ \Comment{...directed path from $xprec_{other}$ to $x_{other}$}
    \If{$r \leq estr$}\label{lcis:notsolo} 
    	\State $view^{r}[other]\gets  H_{other}[estr - r]$
	\State $c \gets 0$ \Comment{proc. $me$ is behind or in the same round as the other proc.}
    \Else{}\label{lcis:solo} 
    	\State $view^{r}[other] = \bot$
	\State $c \gets c + 1$ \Comment{proc. $me$ is ahead.} 
    \EndIf
    \If{$c = \Delta$}\label{lcis:simend} break;   \Comment{quit the simulation after $\Delta$ consecutive solo rounds}\EndIf
    \EndFor
    \State\label{lcis:returnlabel}\textsl{return} \textsc{label}($r,i,view^1,\ldots,view^r)$
    \End
  \end{algorithmic}
\end{algorithm}

\subsection{Correctness of the simulation}

We show that the simulation presented in Algorithm~\ref{alg:constantISsim} is correct. The sequence of views obtained by the processes may have been obtained in an execution of protocol \m{P} in the IS model. We also show that Algorithm~\ref{alg:constantISsim} is able to generate exponentially many IS executions of length $R$. We conclude by proving Theorem~\ref{thm:lcis}, establishing that 2 processes, wait-free fast $\epsilon$-agreement is achievable with constant size registers.

In the following $i$, denote the index of one process and $j$ the index of the other process. 
We consider an execution $e$ of Algorithm~\ref{alg:constantISsim}, which is seen as a sequence of steps $e=s_1,s_2,\ldots$, where each step $s_k$ is a read or write operation performed on one of the register. For $i \in \{1,2\}$, we denote by $read_i^r$ (resp. $write_i^r$) the $r$th read operation (resp. the $r$th write operation) performed by process $i$. Note that the steps $read_i^r$ and $write_i^r$  are performed by process $i$ during its $r$th iteration of the for loop. The successive steps of process $i$ are thus $write_i^1,read_i^1,write_i^2,read_i^2,\ldots$ At process $i$, the $r$th round of the IS execution is simulated when process $i$ performs the $r$th iteration of the for loop. We use the terms iteration and round interchangeably. 

As a process that simulates $\Delta$ consecutive solo rounds exits the simulation, and simulating a round involves a write step and a read steps, a process cannot be too far ``ahead'' of the other process. In other words, in any prefix of $e$, the difference between the number of steps of both processes is bounded. 
\begin{lemma}
  \label{lem:lcis:nottoofarahead}
  Let $r \geq \Delta$ and suppose that process $i$ performs at least $r$ read operations in $e$. $read_i^r$ is preceded by at least $r-\Delta$ write operations by process $j$. 
\end{lemma}
\begin{proof}
  Assume for contradiction that there exists $r \geq \Delta$ such that the $r$th operation by process $i$ is preceded by strictly less than $r-\Delta$ write operations by process $j$. At process $i$, let us consider iteration $r-\delta$ of the for loop, for $\delta \in \{0,\ldots,\Delta\}$. The read $read_i^{r-\delta}$ of $R_j$ performed in this iteration is preceded by at most $r-\Delta-1$ write operations on the same register by process~$j$.

  Theses writes are performed by process $j$ in rounds $1,\ldots, r-\Delta-1$, one write per round (at process $j$,  simulation of round $r$ corresponds to execution of the $r$th  iteration of the for loop.).
  Each  write has as input the current position of process $j$ on the ring, which is incremented at the beginning of every round (line~\ref{lcis:preparex}). Moreover, after reading a position $x$ from $R_j$, process $i$ adds to its estimate $estr$ of the round number of process $j$ the minimal number of rounds required  for process $j$ to move from position $xprec$ to position $x$ on the ring, where $xprec$ is the last known position of process $j$ (line~\ref{lcis:estimateround}). Note that if $d$ is the length of the directed path from node $xprec$ to node $x$,  process $j$ has to perform at least $d$ rounds to move from $xprec$ to $x$ but may have in fact performed $d + k(2\Delta+1)$ rounds, from some $k\geq 0$. That is, process $j$ may have done several laps  without being noticed by process $i$. Therefore, the estimation $estr_i^{r-\delta}$ computed  by process $i$ in round $r-\delta$ is at most $r-\Delta-1$, the number of writes that precedes $read_i ^{r-\delta}$.

Hence, as $\Delta \geq \delta$, we have $r-\delta > r-\Delta-1$, and it follows from the code (line~\ref{lcis:solo}) that   simulated  round $r-\delta$ is solo for process $i$.
  Consequently, process $i$ simulates  $\Delta+1$ consecutive solo rounds $r-\Delta, \ldots,r$, which is not possible, as after $\Delta$ consecutive solo rounds, process $i$ exits the simulation (line~\ref{lcis:simend}). 
\end{proof}

As a corollary of the proof of the previous Lemma, we obtain:
\begin{corollary}
  \label{coro:lcis:quit}
  Let $r \geq \Delta$ and suppose that process $i$ performs at least $r$ reads operations in~$e$. If $read_i^r$ is preceded by exactly $r-\Delta$ writes operations by process $j$, $read_i^r$ is the last step of process $i$ in $e$. 
\end{corollary}
\begin{proof}
  As seen in the proof of Lemma~\ref{lem:lcis:nottoofarahead}, process $i$ simulate rounds $r-\Delta+1,\ldots,r$ as solo rounds. It thus quits the simulation at the end of round $r$, and therefore $read_i^r$ is its last step in $e$. 
\end{proof}

The next Lemma establishes that between two successive read operations by the same process~$i$, the other process $j$ takes a bounded number of steps. Intuitively, this follows from Lemma~\ref{lem:lcis:nottoofarahead} which implies that process $j$ that cannot be too far behind process $i$ in term of number of steps, and implies also that process $j$ cannot perform  too many steps in isolation. Otherwise, process $j$ will simulate $\Delta$ solo steps and will quit the simulation.

\begin{lemma}
  \label{lem:lcis:between}
  Let $k>0$ and suppose that process $i$ performs at least $k+1$ read operations in $e$. The number of write steps by process $j$ between $read_i^k$ and $read_i^{k+1}$ is at most $2\Delta$.
\end{lemma}
\begin{proof}
  The Lemma is trivially true is process $j$ does not take any step between $read_i^k$ and $read_i^{k+1}$. Let us assume that process $j$ writes at least once between  $read_i^k$ and  $read_i^{k+1}$. 

  \begin{itemize}
  \item 
  Let $write_j^\ell$ be the first write by process $j$ that follows $read_i^k$. Since $read_i^k$ is not the last step of process $i$, it follows from Lemma~\ref{lem:lcis:nottoofarahead} and Corollary~\ref{coro:lcis:quit} that $read_i^k$ is preceded by at least $k-\Delta+1$ write steps by process $j$. Therefore, we have $k-\Delta+2 \leq \ell$.
\item
  Let us now consider  $write_j^{\ell'}$  the last write step by process $j$ that precedes $read_i^{k+1}$. If $\ell = \ell'$, the Lemma is true. In the remaining, we suppose that $\ell < \ell'$. By the code, the step of process $j$ preceding $write_j^{\ell'}$ is the read step $read_j^{\ell'-1}$. This step occurs between $read_i^{k}$ and $read_i^{k+1}$. It is thus preceded by at most $k+1$ write steps by process $i$ ($k$ writes before $read_i^k$, and at most one write that occurs between $read_i^k$ and $read_i^{k+1}$).   By Lemma~\ref{lem:lcis:nottoofarahead}, we have $\ell'-1 \leq k+1 + \Delta$. As $read_j^{\ell'-1}$ is not the last step of process $j$, the inequality is strict by Corollary~\ref{coro:lcis:quit}, that is $\ell' \leq k+1 + \Delta$.
\end{itemize}
The number of writes step by process $j$ between $read_i^{k}$ and $read_i^{k+1}$ is $\ell'-\ell+1$
From the two items above, we have
$$\ell' - \ell +1 \leq k+1 + \Delta -(k-\Delta+2) +1 =  2\Delta,$$ which concludes the proof of the Lemma. 
\end{proof}

An immediate consequence of Lemma~\ref{lem:lcis:between} is that process $j$ cannot perform a complete lap of the ring without being noticed by process $i$. 
Indeed, each time process $j$ moves on the ring, it writes its new position to its register $R_j$. Lemma~\ref{lem:lcis:between} establishes that between two reads of $R_j$ by process $i$, at most $2\Delta$ writes are performed by process $j$. Hence, as the size of the ring is $2\Delta+1$, process $j$ cannot do a complete lap between two reads by process $i$. This observation is a key ingredient in the proof of the next Lemma.

\begin{lemma}
  \label{lem:lcis:accurate}
  Let $estr_i^r$ denote the estimation of the round of process $j$ computed by process $i$ in its $r$th iteration of the for loop. If there is no write step by process $j$ preceding $read_i^r$, then $estr_i^r = 0$. Otherwise,  for $write_j^{r'}$ defined as the  last write by process $j$ preceding $read_i^r$, we have $est_i^r = r'$. 
\end{lemma}
\begin{proof}
  We prove the Lemma by induction on the iteration numbers $r$ of the for loop at process~$i$.
  \begin{itemize}
  \item $r=1$. If there is no write by process $j$ before $R_j$ is read by process $i$ in round $1$, process $i$ reads $(0,[\bot,\ldots,\bot])$ and hence $estr_i^1 = 0$ (line~\ref{lcis:estimateround}), as desired. Otherwise, let $r'>0$ be the number of times process $j$ has written to $R_j$ before step $read_i^1$.   As process $i$ performs a single write steps before $read_i^1$, by Lemma~\ref{lem:lcis:nottoofarahead}, $r'$ cannot exceed $\Delta +1 < 2\Delta+1$. Therefore, process $i$ reads $r'$ from $R_j$ in iteration $1$ and we have $estr_i^{1} = r'$. 
  \item Let $r> 1$ and suppose that the Lemma is true for iteration $r-1$. By the induction hypothesis, $estr_i^{r-1} = r'$ where $r'$ is the number of times process $j$ has written before step $read_i^{r-1}$. Let $d$ be the number of times process $j$ writes  between the steps $read_i^{r-1}$ and $read_i^r$. It follows from  Lemma~\ref{lem:lcis:between} that $d < 2\Delta+1$. That is, the new position of process $j$ on the ring as read by process $i$ in step $read_i^r$ is at distance $d$ from its previous known position. From the code (line~\ref{lcis:estimateround}), we have $estr_i^r = est_i^{r-1} + d = r' + d$ and the Lemma follows for iteration $r$ as $d+r'$ is the number of writes by process $j$ preceding $read_i^r$. \qedhere
  \end{itemize}
\end{proof}

\begin{corollary}
  \label{coro:lcis:diff}
If $estr_i^r \geq r$ in round $r \geq 1$ then $est_i^r - r \leq \Delta$. 
\end{corollary}
\begin{proof}
  $estr_i^r$ is the number of writes by process $j$ that precedes the read $read_i^r$ of $R_j$ performed by process $i$ in round $r$ (Lemma~\ref{lem:lcis:accurate}). Let $r' = est_i^r$.
If $r' = 1$, $r=1$ and the Lemma is true. Otherwise, $r'\geq 2$ and the $r'$th write by process $j$ is preceded by  the read step $read_j^{r'-1}$. Since $read_j^{r'-1}$ is not the last step of process $j$, we have that this step is preceded by at least $(r'-1) -\Delta +1 = r'-\Delta$ write steps by process $i$ (Lemma~\ref{lem:lcis:nottoofarahead} and corollary~\ref{coro:lcis:quit}). Each of these writes precede the step $read_i^r$. Therefore, $r \geq r' - \Delta$, from which we conclude that $\Delta \geq r'-r$, that is $\Delta \geq estr_i^r - r$. 
\end{proof}

Let $view_i^r$ denote the $view$ computed by process $i$ in round $r$. In the simulated execution, $view_i^r$ is the snapshot of the $r$th memory $M_r$ obtained by process $i$. Let also $v_i^r$ denote the value written to $M_r[i]$ in the simulated execution (line~\ref{lcis:Pwrite}.). We show that $view$s are valid  and are indeed snapshots:

\begin{lemma}
  \label{lem:lcis:views}
  Suppose that process $i$ completes round $r$ and let $view_i^r$ be the view it gets. We have $view_i^r[i]  = v_i^r$, and if $view_i^r[j] \neq \bot$ then  $view_i^r[j] = v_j^r$. Moreover, if process $j$ also completes round $r$, then there exists $k \in \{1,2\}$ such that $view_k^r = [v_1^r,v_2^r]$. 
\end{lemma}
\begin{proof}
The fact that $view_i^r[i]  = v_i^r$ directly follows from the code (line~\ref{lcis:Pwrite} and line~\ref{lcis:seeme}).

Suppose that $view_i^r[j] \neq \bot$. By the code, $r \leq estr_i^r$ (line~\ref{lcis:notsolo}). Let $r' = estr_i^r$. It follows from lemma~\ref{lem:lcis:accurate} that process $j$ writes to $R_j$ exactly $r'$ times  before $R_j$ is read by process $i$ in round $r$. Hence,  $H_j$, the history reads from $R_j$ in round $r$ by process $j$ is the history written by process $j$ in round $r'$. Therefore, due to how histories are updated (line~\ref{lcis:prepareH}), $H_j = [v_j^{r'}, \ldots, v_j^{r'-k}, \ldots, v_j^{r'-\Delta}]$. By corollary~\ref{coro:lcis:diff}, $r'-r = estr_i^r - r \leq \Delta$ and hence $H_{j}[estr_i^r - r]$ is well defined since $H_j$ has $\Delta+1$ entries indexed $0,\ldots,\Delta$. Finally,  $view_i^r[j] = H_j[estr_i^r - r] = H_j[r' - r] = v_j^{r'-(r'-r)}=v_j^r$, as required. 

We conclude the proof by showing that the simulated round $r$ is solo for at most one process. Suppose that $view_i[j] = \bot$. By the code, $r > estr_i^r = r'$ (line~\ref{lcis:solo}). By Lemma~\ref{lem:lcis:accurate}, process $j$ performs exactly $r'<r$ writes before process $i$ reads $R_j$ in round $r$. That is, when process $i$ reads $R_j$ in round $r$, process $j$ has not started its simulation of round $r$. Hence, the $r$th read of $R_i$ by process $j$ occurs after process $i$ has written $r$ times. Therefore, $r \leq estr_j^r $, from we conclude by the code (line~\ref{lcis:notsolo}) that $view_j^r[i] \neq \bot$. \qedhere
\end{proof}

For a given $R$, in some executions of Algorithm~\ref{alg:constantISsim},  processes may not simulate $R$ rounds of the IS model. This occurs when a process simulates $\Delta$ consecutive solo rounds. Nevertheless, we show that the simulation is able to generate exponentially many IS executions of length $R$. That is, there exists a set of IS executions  $\m{E}$ of size at least $2^R$ such that, for each $E \in \m{E}$:
\begin{itemize}
\item The length of $E$ is $R$:  both processes perform $R$ rounds in $E$ and,
\item There exists an execution $e$ of Algorithm~\ref{alg:constantISsim} in which the simulated IS execution is $E$. 
\end{itemize}

\begin{lemma}
  \label{lem:lcis:exp}
  Let $\Delta \geq 2$ and $R > 0$. There exists $2^R$ IS executions of length $R$ that are simulated by Algorithm~\ref{alg:constantISsim}. 
\end{lemma}
\begin{proof}
  Let $E^r$ denote the prefix of a simulated execution in which both processes have simulated exactly $r$ rounds. We assign to $E^r$ an integer $c(E^r) \in \{0,\ldots,\Delta\}$ which counts the number of consecutive solo rounds by the same process $E^r$ ends with. More precisely, if in the last round both processes have the same view, $c(E^r) = 0$. If, for some $i \in \{1,2\}$, in rounds $r'+1, \ldots, r$, process $i$ is solo but not in round $r'$, $c(E^r) = r-r'$.

  Any prefix $E^r$ with $c(E^r) \in \{0,1\}$ can be extended in two prefixes  $E_0'$ and $E'_1$ in which processes perform one more round, and such that $c(E_0') = 0$ and $c(E'_1) = 1$. As both processes perform $r$ rounds in $E^r$, and at the end of $E^r$ no process has performed $\Delta \geq 2$ consecutive solo rounds, no process quits the simulation at the end of its $r$th iteration of the for loop. Starting from $E^r$, a new round can therefore by simulated. By appropriately scheduling the read and write steps of processes $1$ and $2$, the new round is solo for process $1$, or for process $2$  or for neither of them. For example, suppose that $c(E^r) =1$ and process $i$ is solo in round $r$. $E_0'$ is obtained by first having both processes write to their registers (steps $write_1^{r+1}$ and $write_2^{r+1}$) and then letting them read (steps $read_1^{r+1}$ and $read_2^{r+1}$). The resulting simulated round is solo for neither $p_1$ nor $p_2$. For $E_1'$, we let process $j$ writes and reads before process $i$: the underlying execution of the simulation algorithm is extended with the sequence of steps $$write_j^{r+1},read_j^{r+1},write_i^{r+1},read_i^{r+1}.$$ The resulting simulated round is solo for process $j$, and hence $c(E_1') =1$.
Therefore, starting from an empty execution, one can inductively construct $2^R$ executions of Algorithm~\ref{alg:constantISsim}, each simulating a distinct IS execution in which every process performs $R$ rounds. 
\end{proof}

The next proposition states the correctness of the simulation:
\begin{proposition}
  \label{prop:lcis:sim}
  Let $\Delta \geq 2$ and  $R > 0$, and  $\m{P}$ be a labelling protocol for the IS model.  Algorithm~\ref{alg:constantISsim} simulates executions of $\m{P}$ of length at most $R$ using two registers of size ${O(\log \Delta + b\Delta)}$, where $b$ is the number of bits written in each round by $\m{P}$. The step complexity is $O(R)$ and all executions of the algorithms generate in total $\Omega(2^R)$ distinct IS executions. 
\end{proposition}
\begin{proof}
  By the code, process $i$ writes and reads once to simulate one round of the IS model. The step complexity per process is thus $O(R)$. By Lemma~\ref{lem:lcis:exp}, at least $2^R$ distinct executions are simulated in total by the algorithm. It follows from Lemma~\ref{lem:lcis:views} that each of these executions is a valid IS execution. Each register stores a position on the ring, which requires $\lceil\log(2\Delta+1)\rceil$ bits and an history of  $\Delta+1$ values written by the labelling protocol. Hence the size of each register is   $O(\log \Delta + b\Delta)$.
\end{proof}

\subsection{From simulating IS executions to fast $\epsilon$-agreement}

We finally piece together the findings of this section.

\begin{theorem}
  \label{thm:lcis}
  There exists a wait-free $\epsilon$-agreement protocol for two processes with step complexity $O(\log\frac{1}{\epsilon})$ using two registers  of size~$6$. 
\end{theorem}

\begin{proof}
  Let $\epsilon \in (0,1]$ and let $R$ be an integer. 
  By Lemma~\ref{lem:1bit_disambiguation}, there exists a $2$ processes labelling protocol $\m{P}$ for  the IS model in which one bit is written by each process in each round. The protocol generates $3^{r+1}$ distinct labels over all $r$-rounds executions.
  Proposition~\ref{prop:lcis:sim} shows that one can simulate $\m{P}$ using constant size registers. Simulating $R$ rounds of $\m{P}$ requires  $O(R)$ steps per process, and the total number of distinct IS executions  generated by the simulation is  $\Omega(2^R)$. Since protocol $\m{P}$ assign a distinct label to each distinct local state, the total number of labels generated over all executions of the simulation is $\Omega(2^R)$. 
  Finally, we apply Lemma~\ref{lem:lcis:fromlabeltoeps} to obtain a wait-free $1/2^R$-agreement protocol with step complexity $O(R)$. By choosing $R = O(\log(\frac{1}{\epsilon}))$, the resulting protocol solves $\epsilon$-agreement, uses two register of constant size (in addition to input registers) and has $O(\log\frac{1}{\epsilon})$ complexity.

  In the simulation, the size of each  register is $\left\lceil\log(2\Delta+1)\right\rceil + b(\Delta+1)$, where $b$ is the number of bits communicated by the simulated labelling protocol in each round. 
Setting $\Delta = 2$ in the simulation, and as $\m{P}$ communicate one bit per rounds, two registers of $6$ bits suffice. 
\end{proof}



\section{Conclusion}
\label{sec:concl}

In this paper, we have studied the standard  read/write shared-memory model where $n$ asynchronous crash-prone processes communicate using  one
 single-writer/multi-reader register per process.  For the case where the maximum  number $t$ of crash failures satisfies $1 \leq t < \frac{n}{2}$, we have showed that registers of size $O(t)$ are sufficient. So, for $t=O(1)$, the $t$-resilient model with constant-size registers is universal. If fact, we have shown that even 1-bit registers are universal for the case of a single failure in 2-process systems, using a novel
  2-process approximate agreement algorithm.
 On the other hand, when a majority of the processes can crash, the size of the registers must depend not only on the number $n$ of processes, but also on the task itself. 
This line of work open several avenues for research. In particular, we underline the following open problem: 

Several interesting questions are left for future work, including resolving the case of $t=n/2$,
but perhaps the most significant is  finding
 a  characterization (topological or by other means) of the tasks that are wait-free solvable in the read/write shared-memory model with constant-size registers, or with registers of $f(n)$ bits for some given function~$f$.
%
%

\bibliographystyle{plainurl}
\bibliography{bib}



\end{document}